\documentclass[final,1p,times]{elsarticle}

\usepackage{bbm}
\usepackage{amssymb}
\usepackage{amsthm}
\usepackage{amsmath}
\usepackage{tikz}
\usepackage{siunitx}
\newtheorem{theorem}{Theorem}
\newtheorem{lemma}[theorem]{Lemma}
\usepackage{algorithm,algpseudocode}
\usepackage{enumitem}
\usepackage{lscape}
\usepackage{subfig}
\usepackage[hidelinks]{hyperref}
\usepackage{multirow}

\usepackage{scrextend}
\usepackage[final]{changes}
\definechangesauthor[name={RP}, color=orange]{rp}
\definechangesauthor[name={REV0}, color=blue]{rev0}
\definechangesauthor[name={REV}, color=red]{rev}
\newcommand{\finalstate}[1][]{%
	\renewcommand{\added}[2][]{##2}%
	\renewcommand{\replaced}[3][]{##2}}%

\usepackage{forest}
\usetikzlibrary{arrows.meta}
\tikzset{
	0 my edge/.style={densely dashed, my edge},
	my edge/.style={-{Stealth[]}},
}
\forestset{
	BDT/.style={
		for tree={
			if n children=0{}{circle},
			draw,
			edge={
				my edge,
			},
			if n=1{
				edge+={0 my edge},
			}{},
			font=\sffamily,
		}
	},
}

\journal{Reliability Engineering \& System Safety}

\newcommand{\kterminal}{$\mathcal{K}$-terminal reliability}
\newcommand{\twoterminal}{two-terminal reliability}
\newcommand{\allterminal}{all-terminal reliability}

\newcommand{\instancef}{(G,P)}
\newcommand{\instance}{$\instancef$}
\newcommand{\realf}{G(P)}
\newcommand{\real}{$\realf$}
\newcommand{\unrelf}{u_G(P)}
\newcommand{\unrel}{$\unrelf$}
\newcommand{\unrelef}{\hat{u}_G(P)}
\newcommand{\unrele}{$\unrelef$}
\newcommand{\relf}{r_G(P)}
\newcommand{\rel}{$\relf$}

\newcommand{\hll}{HLL}
\newcommand{\aaaf}{\mathcal{AA}}
\newcommand{\aaa}{$\aaaf$}
\newcommand{\sraf}{\textit{SRA}}
\newcommand{\sra}{$\sraf$}
\newcommand{\gbasf}{\textit{GBAS}}
\newcommand{\gbas}{$\gbasf$}
\newcommand{\relnet}{$\mathcal{K}$-RelNet}

\newcommand{\epsof}{\epsilon_{\text{o}}}
\newcommand{\deltaof}{\delta_{\text{o}}}
\newcommand{\epso}{$\epsof$}
\newcommand{\deltao}{$\deltaof$}

\newcommand{\erf}{\text{ER}}

\newcommand{\Expect}{E}

\begin{document}

\begin{frontmatter}



\title{Principled Network Reliability Approximation: A Counting-Based Approach\footnote{Accepted. DOI: \href{https://doi.org/10.1016/j.ress.2019.04.025}{10.1016/j.ress.2019.04.025}. {\textcopyright}2019.  Licensed under the Creative Commons \href{https://creativecommons.org/licenses/by-nc-nd/2.0/}{BY-NC-ND}.}}

\cortext[cor1]{E-mail: roger.paredes@rice.edu}
\author[label1]{R. ~Paredes\corref{cor1}}
\author[label1]{L. Due\~nas-Osorio}
\author[label2,label3]{K.S. Meel}
\author[label3]{M.Y. Vardi}
\address[label1]{Department of Civil and Environmental Engineering, Rice University, Houston, USA.}
\address[label2]{Computer Science Department, National University of Singapore, Singapore.}
\address[label3]{Department of Computer Science, Rice University, Houston, USA.}

\begin{abstract}
	As engineered systems expand, become more interdependent, and operate in real-time, reliability assessment is key to inform investment and decision making. However, network reliability problems are known to be \#P-complete, a computational complexity class believed to be intractable, and thus motivate the quest for approximations. Based on their theoretical foundations, reliability evaluation methods can be grouped as: (\romannumeral 1) \textit{exact or bounds}, (\romannumeral 2)\textit{ guarantee-less sampling}, and (\romannumeral 3) \textit{probably approximately correct} (PAC). Group (\romannumeral 1) is well regarded due to its useful byproducts, but it does not scale in practice. Group (\romannumeral 2) scales well and verifies desirable properties, such as the \textit{bounded relative error}, but it lacks error guarantees. Group (\romannumeral 3) is of great interest when precision and scalability are required. We introduce {\relnet}, an extended counting-based method that delivers PAC guarantees for the {\kterminal} problem. We also put our developments in context relative to classical and emerging techniques to facilitate dissemination. Then, we test in a fair way the performance of competitive methods using various benchmark systems. We note the range of application of algorithms and suggest a foundation for future computational reliability and resilience engineering, given the need for principled uncertainty quantification in complex systems.
\end{abstract}


\end{frontmatter}


\section{Introduction}

\noindent Modern societies rely on physical and technological networks such as transportation, power, water, and telecommunication systems. Quantifying their reliability is imperative in design, operation, and resilience enhancement. Typically, networks are modeled using a graph where vertices and edges represent unreliable components. Network reliability problems ask: what is the probability that a complex system with unreliable components will work as intended under prescribed functionality conditions?

In this paper, we focus on the {\kterminal} problem~\cite{Ball1986}. In particular, we consider an undirected graph $G=(V,E,\mathcal{K})$, where $V$ is the set of vertices, $E\subseteq V\times V$ is the set of edges, and $\mathcal{K}\subseteq V$ is the set of terminals. We let $G(P)$ be a stochastic graph, where every edge $e\in E$ vanishes from $G$ with respective probabilities $P=(p_e)_{e\in E}$. We assume a binary-system, and say {\real} is \textit{unsafe} if a subset of vertices in $\mathcal{K}$ becomes disconnected, and \textit{safe} otherwise. Thus, given instance {\instance} of the {\kterminal} problem, we are interested in computing the unreliability of $G(P)$, denoted {\unrel}, and defined as the probability that $G(P)$ is unsafe.

If $|\Theta|$ is the cardinality of set $\Theta$, then $n=|V|$ and $m=|E|$ are the number of vertices and edges, respectively. Also, when $|\mathcal{K}|=n$ and $|\mathcal{K}|=2$, the {\kterminal} problem reduces to the all-terminal and two-terminal reliability problems, respectively. These are well-known and proven to be~\#P-complete problems~\cite{Ball1986,Valiant1979}. The more general {\kterminal} problem is \#P-hard, so ongoing efforts to compute {\unrel} focus on practical bounds and approximations. 

Exact and bounding methods are limited to networks of small size, or with bounded graph properties such as treewidth and diameter~\cite{Hardy2007, Canale2016}. Thus, for large  $G$ of general structure, researchers and practitioners lean on simulation-based estimates with acceptable Monte Carlo error~\cite{Fishman1986}. However, in the absence of an error prescription, simulation applications can use unjustified sample sizes and lack a priori rigor on the quality of the estimates, thus becoming \textit{guarantee-less} methods.

A formal approach to guarantee quality in Monte Carlo applications relies on the so-called $(\epsilon,\delta)$  \textit{approximations}, where $\epsilon$ and $\delta$ are user specified parameters regarding the relative error and confidence, respectively. As an illustration, for $Y$ as a random variable (RV), say we are interested in computing its expected value $\Expect[Y]=\mu_Y$. Then, after \textit{we specify} parameters $\epsilon,\delta\in(0,1)$, an $(\epsilon,\delta)$ approximation returns estimate $\overline{\mu}_Y$ such that $\Pr(|\overline{\mu}_Y/\mu_Y-1|\geq \epsilon)\leq \delta$. In other words, an $(\epsilon,\delta)$  \textit{approximation} returns an estimate with relative error below $\epsilon$ with at least confidence $1-\delta$.
We term \textit{Probably Approximately Correct} (PAC) the family of methods whose algorithmic procedures deliver estimates with $(\epsilon,\delta)$ guarantees.\footnote{We borrow the PAC terminology from the field of artificial intelligence~\cite{Valiant1984}.}

Having a formal notion of error, we can rigorously address a key issue in Monte Carlo applications: the \textit{sample size}, herein denoted $N$. Using standard probability arguments, and positive finite $\mu_Y$ as the only assumption,\footnote{In this paper $\mu_Y$ will be a probability, such as network unreliability {\unrel}.} we derive: $N=O(\sigma^2_Y/\mu^2_Y\epsilon^{-2}\log1/\delta)$ (See appendix, Theorem~\ref{theom:markov}), exposing Monte Carlo's weakness when required to guarantee results. To make it self-evident, let us model our binary-system as $Y$, a Bernoulli RV, such that $\mu_Y=\unrelf$. Then, note that the substitution of $\mu_Y$ and $\sigma^2_Y$ in $N$ leads to $N\propto 1/\unrelf$, which can be prohibitively large as engineered systems are \textit{highly-reliable} by design.

The sample size issue is a well researched subject of \textit{rare-event} simulation, and we refer readers to Chapter 2~\cite{Fishman1996} and Chapter 1~\cite{Gertsbakh2016} for more background. Attempts to make simulation more affordable include: the Multilevel Splitting method~\cite{Botev2012,Zuev2015}, the recursion-based Importance Sampling method~\cite{Cancela1995}, the Permutation Monte Carlo-based method~\cite{Gertsbakh2010} and its Splitting Sequential Monte Carlo extension~\cite{Vaisman2016}, among others. Some of these techniques verify desired properties, such as the Bounded Relative Variance (BRV) or $\sigma_Y^2/\mu_Y^2=O(1)$, and the Vanishing Relative Variance (VRV) or $\sigma_Y^2/\mu_Y^2=o(1)$, where $Y$ denotes the Monte Carlo estimate returned by a sampling technique.\footnote{Where the little-o notation $f(n)=o(g(n))$ stands for $f/g \rightarrow 0$, for $n>n_0$, $n,n_0\in \mathbb{N}$.} Despite being effective in the rare-event setting, these methods often appeal to the central limit theorem and do not assure quality of error or performance, thus remaining guarantee-less to users.

Naturally, a method that overcomes the rare-event issue while delivering rigorous error guarantees would be of great use in reliability applications. In other words, system reliability is calling for \textit{efficient} PAC methods for a rigorous treatment of uncertainties. Theoretically speaking, an efficient method runs in polynomial time as a function of the size of $(G,P)$, $1/\epsilon$, and $\log(1/\delta)$. In the computer science literature, such a routine is called a \textit{fully polynomial randomized approximation scheme} (FPRAS) for network unreliability. Clearly, efficient in theory does not imply efficient in practice, e.g., the order of the polynomial function bounding the worst-time complexity can be arbitrarily large. Thus, it is imperative to complement theoretically sound developments with computer evaluations. To the best of our knowledge, there is no known FPRAS for the {\kterminal} problem. However, there is a precedent, where Karger gave the first FPRAS for the all-terminal reliability case~\cite{Karger2001}.

To tackle computational and precision issues, this paper develops {\relnet}, a counting-based PAC method for network unreliability that inherits properties of state-of-the-art approximate model counters in the field of computational logic~\cite{Meel2017}. Our approach delivers rigorous $(\epsilon,\delta)$ guarantees and is \textit{efficient} when given access to an NP-oracle: a black-box that solves nondeterministic polynomial time decision problems. The use of NP-oracles for randomized
approximations, first proposed by Stockmeyer~\cite{Stockmeyer1983}, is increasingly within reach as in practice we can leverage efficient solvers for Boolean satisfiability (SAT) that are under active development. Given the variety of methods to compute {\unrel}, we showcase our developments against alternative approaches. In the process, we highlight methodological connections missed in the engineering reliability literature, key theoretical properties of our method, and unveil practical performance through fair computational experiments by using existing and our own benchmarks.

The rest of the manuscript is structured as follows. Section~2 gives background on network reliability evaluation and its $(\epsilon,\delta)$ approximation, as well as the necessary background on Boolean logic before introducing our new counting-based approach: {\relnet}, an efficient PAC method for the {\kterminal} problem. Section~3 contextualizes our contribution relative to other techniques for network reliability evaluation. We highlight key properties for users and draw important connections in the literature. Section~4 presents the main results of our computational evaluation. Section~5 rounds up this study with conclusions and promising research directions.

\section{Counting-Based Network Reliability Evaluation}
\noindent We begin this section with relevant mathematical background and notation, then we introduce the new method, termed {\relnet}. We do so through a fully worked out example for counting-based reliability estimation.
\subsection{Principled network reliability approximation}
\noindent Given instance {\instance} of the {\kterminal} problem, we represent a realization of the stochastic graph {\real} as an $m$-bit vector $X=(x_e)_{e\in E}$, with $m = |E|$, such that $x_{e}=0$ if edge $e\in E$ is failed, and $x_e=1$ otherwise. Note that $\Pr(x_{e}=0)=p_{e}$, and that the set of possible realizations is $\Omega=\{0,1\}^m$. Furthermore, let $\Phi:\Omega\mapsto\{0,1\}$ be a function such that $\Phi(X)=0$ if some subset of $\mathcal{K}$ becomes disconnected, i.e. $X$ is \textit{unsafe}, and $\Phi(X)=1$ otherwise. Also, we define the \textit{failure} and \textit{safe} domains as $\Omega_f=\{X\in\Omega: \Phi(X)=0 \}$ and $\Omega_s=\{X\in\Omega: \Phi(X)=1 \}$, respectively. In practice, we can evaluate $\Phi$ efficiently using breadth-first-search. 

Network reliability, denoted as {\rel}, can be computed as follows:
\begin{equation} \label{eq:brutef}
\relf = 1 - \unrelf =\sum_{X\in \Omega}  \Phi(X)\cdot\Pr(X) ,
\end{equation}
\begin{equation}\label{eq:probx}
	\Pr(X) = \prod_{e_i\in E}^{}p_{e_i}^{(1-x_{i})}\cdot(1-p_{e_i})^{x_{i}},
\end{equation}
where Eq.~(\ref{eq:probx}) assumes independent edge failures. Clearly, the number of terms $|\Omega|=2^m$ of Eq.~(\ref{eq:brutef}) grows exponentially, rendering the brute-force approach useless in practice, and motivating the development of network reliability evaluation methods that can be grouped into: exact or bounds, guarantee-less simulation, and probably approximately correct (PAC).

When exact methods fail to scale in reliability calculations, simulation is the preferred alternative. However, mainstream applications of simulation lack performance guarantees on error and computational cost. Typically, users embark on a trial and error process for choosing the sample size, trying to meet, if at all possible, a target empirical measure of variance such as the coefficient of variation. However, similar approaches have been shown to be unreliable~\cite{Bayer2014}, jeopardizing reliability applications at a time when uncertainty quantification is key, as systems are increasingly complex~\cite{Ellingwood2016}.

To secure a rigorous application of the Monte Carlo method, we use $(\epsilon,\delta)$ approximation methods, which use no assumptions such as the central limit theorem, and that give guarantees of approximation in the non-asymptotic regime, i.e., they deliver \textit{provably} sound approximations with a finite number of samples. Formally, for input parameters $\epsilon,\delta\in (0,1)$, we define a PAC method for network unreliability evaluation as one that outputs estimate {\unrele} such that:
\begin{equation}\label{eq:pac}
\Pr\bigg(\frac{|\unrelef-\unrelf|}{\unrelf}\geq \epsilon\bigg) \leq \delta.
\end{equation}

Recently, the authors introduced RelNet~\cite{Duenas-Osorio2017}, a counting based framework for approximating the {\twoterminal} problem that issues $(\epsilon,\delta)$ guarantees. In this paper, we introduce {\relnet}, an extension that, to the best of our knowledge, is the first \textit{efficient} PAC method for the general {\kterminal} problem.

Next, we survey important background in Boolean logic definitions before introducing {\relnet}.
\subsection{Boolean logic}
\noindent A Boolean formula $\psi:X \in\{0,1\}^n \to \{0,1\}$ is in conjunctive normal form (CNF) when written as $\psi(X)=C_1\wedge \cdots \wedge C_m$, with each clause $C_i$ a disjunction of literals, e.g., $C_1=x_1 \vee \neg x_2 \vee x_3$. We are interested in solving the $\#$SAT (``Sharp SAT'') problem, which counts the number of variable assignments satisfying a CNF formula. Formally, $\#\psi = \big|\big\{X\in \{0,1\}^n | \psi(X)=1\big\}\big|$.
For example, consider the expression $x_1 \neq x_2$. Its CNF representation is $\psi(X) = (x_1 \vee x_2) \wedge (\neg x_1 \vee \neg x_2)$, and the number of satisfying assignments of $\psi$ is $\#\psi=2$.

Furthermore, for Boolean vectors of variables $X=(x_1,\dots,x_n)$ and $S=(s_1,\dots,s_p)$, define a $\Sigma_1^1$ formula as one that is expressed in the form $F(X,S)=\exists S[ \psi(X,S)]$, with $\psi$ a CNF formula over variables $X$ and $S$. Similarly, we are interested in its associated counting problem, called projected counting or ``$\#\exists$SAT.'' Formally, $\#F = \big|\big\{X\in \{0,1\}^n |\exists S\text{ such that }\psi(X,S)=1\big\}\big|$. We use $\Sigma_1^1$ formulas because they let us introduce needed auxiliary variables ($S$) for global-level Boolean constraints, such as reliability, but count strictly over the problem variables ($X$). As an example, consider the expression $[(x_1 \text{ OR } s_1) \neq x_2]$. Its CNF representation is $\psi(X, S) = (x_1 \vee s_1 \vee x_2) \wedge (\neg x_1 \vee \neg s_1 \vee \neg x_2)$, and note the difference between the associated counts $\#\psi=6$ and $\#F=4$. The latter is smaller because the quantifier $\exists$ over variables $S$ ``projects'' the count over variables $X$. To better grasp this projection, observe that $F(X,S)=\exists S [\psi(X,S)]$ is \textit{equivalent} to $\bigvee_{S\in\{0,1\}^p}[\psi(X, S)]$, which in our example simplifies to $(\neg x_1 \vee \neg x_2) \vee  (x_1 \vee x_2)=1$, i.e., for every assignment of variables $X\in\{0,1\}^2$, there is $S\in\{0,1\}$ such that $F(X,S)=1$, and thus $\#F=4$. The equivalent form is shown only for illustration purposes, as it is intractable to work with due to its  length growing exponentially in the number of variables in $S$. Instead, we feed $F(X,S)=\exists S [\psi(X,S)]$ to a state-of-the-art approximate model counter~\cite{SM19}.

Next, we introduce $F_{\mathcal{K}}$, a $\Sigma_1^1$ formula encoding the unsafe property of a graph $G$, and show that $\#F_{\mathcal{K}}=|\Omega_f|$. Recall $\Omega_f$ is the network failure domain $\Omega_f=\{X\in\Omega: \Phi(X)=0 \}$. Moreover, using a polynomial-time reduction to address arbitrary edge failure probabilities, we solve the {\kterminal} problem by computing $\#F_{\mathcal{K}}$. The problem of counting the number of satisfying assignments of a Boolean formula is hard in general, but it can be approximated efficiently via state-of-the-art PAC counters with access to an NP-oracle. In practice, an NP-oracle is a SAT solver capable of handling formulas with up to million a variables, which is orders of magnitude larger than typical network reliability instances.

\begin{figure}
	\centering
	\subfloat[Weighted instance.\label{subf:weight}]{
		\begin{tikzpicture}[darkstyle/.style={circle,draw,fill=gray!65,minimum size=0.4cm}]
		\node [darkstyle] (1) at (0,0) {\scriptsize a};
		\node [circle, draw] (2) at (2,2) {\scriptsize b};
		\node [circle, draw] (3) at (2,-2) {\scriptsize c};
		\node [darkstyle] (4) at (4,0) {\scriptsize d};
		\draw (1)--(2) node[midway, above] {$p_{e_1}=1/2$};
		\draw (1)--(3) node[midway, above] {$p_{e_2}=3/8$};
		\draw (2)--(4) node[midway, above] {$p_{e_3}=1/2$};
		\draw (3)--(4) node[midway, above] {$p_{e_4}=1/2$};
		\node [] (aux) at (2,0) {$(G,P)$};
		\end{tikzpicture}
	}
	\subfloat[Unweighted instance.\label{subf:unweight}]{
		\begin{tikzpicture}[darkstyle/.style={circle,draw,fill=gray!65,minimum size=0.4cm}]
		\pgfmathsetmacro{\dx}{0}
		\node [darkstyle] (1) at (0+\dx,0) {\scriptsize a};
		\node [circle, draw] (2) at (2+\dx,2) {\scriptsize b};
		\node [circle, draw] (13) at (1+\dx-.5,-1-.5) {\scriptsize $v_1$};
		\node [circle, draw] (3) at (2+\dx,-2) {\scriptsize c};
		\node [darkstyle] (4) at (4+\dx,0) {\scriptsize d};
		\draw (1)--(2) node[midway, above] {$e_1$};
		\draw (1)--(3) node[midway, above] {$e_2$};
		\draw (2)--(4) node[midway, above] {$e_3$};
		\draw (3)--(4) node[midway, above] {$e_4$};
		\draw (1)--(13) node[midway, left] {$e_5$};
		\draw (13)--(3) node[midway, below] {$e_6$};
		\node [] (aux) at (2+\dx,0) {$(G',P_{1/2})$};
		\end{tikzpicture}}\\
	\subfloat[Terms of $\Sigma_1^1$ formula $F_{\mathcal{K}}$ and exact counting calculations.\label{subf:fk}]{
		\begin{tikzpicture}[darkstyle/.style={circle,draw,fill=gray!65,minimum size=0.4cm}]
		\pgfmathsetmacro{\dx}{0}
		\node[] (tk) at (5,-4) {\small \begin{tabular}{l}
			\small $C_{e_1}=(s_a \wedge x_{e_1} \rightarrow s_b) \wedge (s_b \wedge x_{e_1} \rightarrow s_a)$,
			\small $C_{e_2}=(s_a \wedge x_{e_2} \rightarrow s_c) \wedge (s_c \wedge x_{e_2} \rightarrow s_a)$.\\
			\small $C_{e_3}=(s_b \wedge x_{e_3} \rightarrow s_d) \wedge (s_d \wedge x_{e_3} \rightarrow s_b)$,
			\small $C_{e_4}=(s_c \wedge x_{e_4} \rightarrow s_d) \wedge (s_d \wedge x_{e_4} \rightarrow s_c)$.\\
			\small $C_{e_5}=(s_a \wedge x_{e_5} \rightarrow s_{v_1}) \wedge (s_{v_1} \wedge x_{e_5} \rightarrow s_a)$,
			\small $C_{e_6}=(s_{v_1} \wedge x_{e_6} \rightarrow s_c) \wedge (s_c \wedge x_{e_6} \rightarrow s_{v_1})$.\\
			\multicolumn{1}{l}{\small $S=\{s_a, s_b, s_c, s_d, s_v\}$, $F_{\mathcal{K}}=\exists S \big((s_a\vee s_d) \wedge (\neg s_a\vee \neg s_d) \wedge \bigwedge_{i=1}^{6} C_{e_i} \big)$.}\\
			\multicolumn{1}{l}{    \small $\#F_{\mathcal{K}}=33$, $\unrelf=\#F_{\mathcal{K}}/2^{|E'|}=33/64$.}
			\end{tabular}
		};
		\end{tikzpicture}
	}
	\caption{{\relnet} example with $\mathcal{K}$$=$$\{\text{a},\text{b}\}$. (a) Original instance, (b) its reduction to $p_e=1/2$, $\forall e \in E'$, and (c) exact counting $\#F_{\mathcal{K}}$.}
	\label{fig:relnet}
\end{figure}

\subsection{Reducing network reliability to counting}
\noindent Next we introduce the {\relnet} formulation. Given propositional variables $S=(s_u)_{u\in V}$ and propositional variables $X=(x_{e})_{e\in E}$, then define:
\begin{equation}\label{eq-tran}
C_{e} = \big[(s_u \land x_{e}) \rightarrow s_v\big] \wedge \big[(s_v \land x_{e}) \rightarrow s_u\big] , \forall e\in E,
\end{equation}
\begin{equation}\label{eq-relnet}
F_\mathcal{K} =\exists S[\psi(X,S)]   =\exists S\bigg[\bigg(\bigvee_{j\in \mathcal{K}} s_j\bigg) \land \bigg(\bigvee_{k\in \mathcal{K}} \neg s_k\bigg) \land \bigwedge_{e\in E} C_{e}\bigg],
\end{equation}
where in Eq.~\ref{eq-tran}, each edge $e\in E$ has end vertices $u,v\in V$. Propositional edge variable $x_{e}$ encodes the state of edge $e\in E$, such that $x_{e}$ is true iff $e$ is not failed, which is consistent with the representation of a realization of the stochastic graph $G(P)$ introduced earlier. An example of $F_{\mathcal{K}}$ is given in Figure~\ref{subf:unweight}-\ref{subf:fk}. Note that $F_{\mathcal{K}}$ is a $\Sigma^1_1$ formula,\footnote{Use identity $(a \wedge b) \rightarrow c \equiv \neg a \vee \neg b \vee c$ for constraints $C_e$ in Eq.~(\ref{eq-tran}).} and we define its associated set of satisfying assignments as $R_{F_\mathcal{K}}= \{X\in \Omega | (\exists S)\psi(X,S)=1\big\}$, such that $\#F_{\mathcal{K}}=|R_{F_\mathcal{K}}|$. Also, recall that the notation for the complement of set $\Theta$ is $\overline{\Theta}$. The next Lemma proves the core result of our reduction.
\begin{lemma}\label{lm:pathcount}
	For a graph $G$ $=$ $(V,E,\mathcal{K})$, edge failure probabilities $P=(p_{e})_{e\in E}$, and $F_{\mathcal{K}}$ and $\Omega_f$ as defined above, we have $\#F_{\mathcal{K}}=|\Omega_f|$. Moreover, for $P_{1/2} = (1/2)_{e \in E}$, we have
	\begin{equation*}
		u_G(P_{1/2})=\frac{\#F_{\mathcal{K}}}{2^{|E|}}.
	\end{equation*}
\end{lemma}
\begin{proof}
	We use ideas from our previous work~\cite{Duenas-Osorio2017}, which deals with the special case $|\mathcal{K}|=2$. First, note that for sets $A$ and $B$ such that $|A| + |\overline{A}| = |B| + |\overline{B}|$, we have $|A|=|B|$ iff there is a bijective mapping from $\overline{A}$ to $\overline{B}$. Moreover, the number of unquantified variables in Eq.~(\ref{eq-relnet}) is $|E|$, so we can establish the next equivalence between the number of distinct edge variable assignments and system states: $|R_{F_{\mathcal{K}}}|+|\overline{R_{F_{\mathcal{K}}}}|=|\Omega_f|+|\overline{\Omega_f}|=2^{|E|}$. Next, we prove $X\in \overline{R_{F_{\mathcal{K}}}} \iff X \in \overline{\Omega_f}, \forall X\in \{0,1\}^{|E|}$, via a bijective mapping.
	
	1) Case $\overline{\Omega_f} \to \overline{R_{F_{\mathcal{K}}}}$: assume $X\in\overline{\Omega_f}$, i.e. $\phi(X)=1$ or $G$ is $\mathcal{K}$-connected. Next, we show that $X\in \overline{R_{F_{\mathcal{K}}}}$, i.e., $F_{\mathcal{K}}(X,S)$ evaluates to false for all possible assignments of variables $S$, due to Eqs.~\ref{eq-tran}-\ref{eq-relnet}. We show this by way of contradiction. Assume there is an assignment $S\in\{0,1\}^{|V|}$ such that $F_{\mathcal{K}}(X,S)$ is true. We deduce this happens iff (i) $\exists j,k\in\mathcal{K}$ such that $s_j\neq s_k$, from Eq.~(\ref{eq-relnet}), and (ii) for every edge $e\in E$ with end-vertices $u,v\in V$ we have $s_v=1$ (resp. $s_u=1$) whenever $x_e$  and $s_u$ (resp. $s_v$) are equal to $1$, due to clause $C_{e}$ in Eq.~(\ref{eq-tran}). Without loss of generality, we satisfy condition (i) setting $s_j=1$ and $s_k=0$, with $j,k\in\mathcal{K}$. Recall $X\in\overline{\Omega_f}$, i.e. $\phi(X)=1$, so there is a path $P=\{j,\dots,k\}\subseteq V$ connecting vertices $j,k\in\mathcal{K}$ and traversing edges $T\subseteq E$ such that $x_e =1,\forall e\in T$. By iterating over constraints $C_{e}$, $\forall e\in T$, and since $s_j=1$, we are forced to assign $s_i=1,\forall i\in P$, to satisfy condition (ii). This assignment results into $s_k=1$, which contradicts condition (i) when we have set $s_k=0$ at the beginning. Thus, an $S\in\{0,1\}^{|V|}$ such that $F_{\mathcal{K}}(X,S)$ is true does not exists, and $X \in  \overline{R_{F_{\mathcal{K}}}}$.
	
	2) Case $\overline{R_{F_{\mathcal{K}}}} \to \overline{\Omega_f}$: assume $X\in\overline{R_{F_{\mathcal{K}}}}$, i.e. $F_{\mathcal{K}}(X,S)$ is false, to show that $X\in\overline{\Omega_f}$. Again, by way of contradiction, we assume $X\in\Omega_f$ and using the arguments from above we deduce that the set of edges $T=\{e\in E| x_e=1\}$ connects every pair of vertices $i,j\in \mathcal{K}$, i.e. $\phi(X)=1$ by definition of $\phi$. This contradicts the definition $\Omega_F=\{X\in \{0,1\}^{|E|}| \phi(X)=0\}$. Thus, we conclude $X\in \overline{\Omega_f}$.
	
	Since we established a bijective mapping between $\overline{R_{F_{\mathcal{K}}}}$ and $\overline{\Omega_f}$, we conclude $\#F_{\mathcal{K}} = |\Omega_f|$.
	The last part of the lemma follows by noting that $\Pr(X)=1/2^{|E|}$ when $P=P_{1/2}$, so that $\unrelf=\sum_{x\in\Omega} (1-\Phi(X))\cdot\Pr(X)= |\Omega_f|\cdot 1/2^{|E|}=\#F_{\mathcal{K}}/2^{|E|}$.
\end{proof}
Now we generalize $u_G(P_{1/2}) =\#F_{\mathcal{K}}/2^{|E|}$ to arbitrary edge failure probabilities. To this end, we use a weighted to unweighted transformation~\cite{Duenas-Osorio2017}.
\subsection{Addressing arbitrary edge failure probabilities}
\noindent The next definitions will be useful for stating our weighted-to-unweighted transformation. Let $0.b_1\cdots b_m$ be the binary representation of probability $q\in(0,1)$, i.e. $q=\sum_{k=1}^{m}b_k/2^k$. Define $z_k$ ($\bar{z}_k$) as the number of zeros (ones) in the first $k$ decimal bits of the binary representation. Formally, $z_k=k-\sum_{i=1}^{k}b_i$ and $\bar{z}_k=k-z_k$, $\forall k\in L$, with $L=\{1,\dots,m\}$. Moreover, for $V=\{v_0,\dots,v_{z_m+1}\}$, define a function $\eta:L\to V\times V$ such that $\eta(k)=(v_{z_{k-1}},v_{z_k})$ if $b_k=0$, and $\eta(k)=(v_{z_{k-1}},v_{z_m+1})$ otherwise. We will show that, for $E=\bigcup_{k\in L}\eta(k)$ and $\mathcal{K}=\{v_0,v_{z_m+1}\}$, $G(V,E,\mathcal{K})$ is a series-parallel graph such that $r_G(P_{1/2})=q$. Thus, our weighted-to-unweighted transformation entails replacing every edge $e\in E$ with failure probability different from 1/2 with a reliability preserving series-parallel graph $G^e$.

For example, from Figure~\ref{subf:weight}, the binary representation of $1-p_e=5/8$ is 0.101, so we have $m=3$, $z_m=1$, and $\bar{z}_m=2$. Also,  we replace edge $e_2$ with a series parallel graph $G^{e_2}$ using the construction from above, which yields $V^{e_2}=\{v_0, v_1, v_2\}$, $E^{e_2}=\{(v_0, v_{2}),(v_0, v_1),(v_1, v_2)\}$, and terminal set $\mathcal{K}^{e_2}=\{v_0, v_2\}$. Since $u_{G^{e_2}}(P_{1/2})=3/8$, we replace $e_2$ by $G^{e_2}$ as shown in Figure~\ref{subf:unweight}, where $v_0=a$ and $v_2=c$, for consistency with the global labeling of the figure. The next lemma proves the correctness of this transformation.
\begin{lemma} \label{lm:chaingraph}
	Given probability $q=0.b_1\cdots b_m$ in binary form, graph $G=(V,E,\mathcal{K})$ such that $V=\{v_0,\dots,v_{z_m+1}\}$, $E=\{\eta(1),\dots,\eta(m)\}$ and $\mathcal{K}=\{v_0, v_{z_m+1}\}$, and $P_{1/2}=(1/2)_{e\in E}$, we have $r_{G}(P_{1/2})=q$ and $|V|+|E|=z_{m}+2+m$.
\end{lemma}
\begin{proof}
		Define $G_k=(V_k,E_k), \forall k\in L$, with $E_k=\{\eta(1),\dots,\eta(k)\}$ and $V_k=\cup_{i=1}^{k}\{v_j:v_j\in\eta(i)\}$. Clearly, $V=V_m$ and $E=E_m$. The key observation is that $G$ is a series-parallel graph and that we can enumerate all paths from $v_0$ to $v_{z_m+1}$ in $G$. Let $k_1=\min\{k\in L:b_{k_1}=1\}$. Then, the edge set $E_{T_1}=E_{k_1}$ forms a path from $v_0$ to $v_{z_m+1}$, denoted $T_1$, with vertex sequence $(v_0,\dots,v_{z_{k_1}},v_{z_m+1})$, size $|E_{T_{1}}|=z_{k_1}+1$, and $\Pr(T_1)=1/2^{z_{k_1}+1}$. Next, for $k_2$ the second smallest element of $L$ such that $b_{k_2}=1$, $G_{k_2}$ contains a total of two paths, $T_1$ and $T_2$, with $T_1$ as before and $E_{T_2}=E_{k_2} \setminus \{(v_{z_{k_1}}, v_{m+1})\}$ of size $z_{k_2}+1$. Also, $E_{k_2}=E_{T_1}\cup E_{T_2}$ and $E_{T_1}\cap E_{T_2}=E_{T_1}\setminus \{(v_{z_{k_1}}, v_{m+1})\}$. Thus, the event $\overline{T}_1T_2$ happens iff edge $(v_{z_{k_1}}, v_{m+1})$ fails and edges in $E_{T_2}$ do not fail, letting us write $\Pr(\overline{T}_1T_2)=1/2\cdot 1/2^{z_{k_2}+1}$. For $k_j$ the $j$-th smallest element of $L$ such that $b_{k_j}=1$, $G_{k_j}$ has a total of $j=\bar{z}_{k_j}$ paths, with $E_{T_j}=E_{k_j} \setminus\cup_{i=1}^{j-1} \{(v_{z_{k_i}}, v_{m+1})\}$, $|E_{T_j}|=z_{k_k}+1$, and $E_{k_j}=\cup_{i=1}^{j}E_{T_i}$. Furthermore, event $\overline{T}_1\cdots \overline{T}_{j-1}T_{j}$ happens iff edges in $\cup_{i=1}^{j-1} \{(v_{z_{k_{i}}}, v_{m+1})\}$ fail and edges in $E_{T_j}$ do not fail. Thus,  $\Pr(\overline{T}_1\cdots \overline{T}_{j-1}T_{j})=1/2^{\bar{z}_{k_j}-1}\cdot 1/2^{z_{k_j}+1}=1/2^{k_j}$.This leads to $\relf=\Pr(T_1)+\Pr(\overline{T}_1T_2)+\dots+\Pr(\overline{T}_1\cdots \overline{T}_{\bar{z}_m-1}S_{\bar{z}_m})=\sum_{i=1}^{\bar{z}_m}1/2^{k_i}$. Rewriting the summation over all $k\in L$ yields $\relf=\sum_{k=1}^{m}b_k/2^{k}$, which is the decimal form of $q=0.b_k\cdots b_m$. Furthermore, $|V|=z_m+2$ and $|E|=m$ from their definitions.
\end{proof}
Now we leverage Lemma~\ref{lm:chaingraph} to introduce our general counting-based algorithm for the {\kterminal} problem.
\subsection{The new algorithm: {\relnet}}
{\relnet} is presented in Algorithm~\ref{alg:relnet}. Theorem~\ref{th:3} proves its correctness. Figure~\ref{fig:relnet} illustrates the exact version beginning with the reduction to failure probabilities of 1/2, and rounding up with the construction of $F_\mathcal{K}$ and exact counting of its satisfying assignments. In Algorithm~\ref{alg:relnet}, however, we use an approximate counter giving $(\epsilon,\delta)$ guarantees~\cite[][Chapter 4]{Meel2017}.
\begin{theorem}\label{th:3}
	Given an instance {\instance} of the {\kterminal} problem and $M$ defined as in Algorithm~\ref{alg:relnet}:
	\begin{equation*}
	\unrelf=\#F_{\mathcal{K}}/2^M.
	\end{equation*}
\end{theorem}
\begin{proof}
	The proof follows directly from Lemmas~\ref{lm:pathcount} and~\ref{lm:chaingraph}. First, note that the transformation in step 1 of {\relnet} outputs an instance ($G',P_{1/2}$) so that $\unrelf=u_{G'}(P_{1/2})$, where $P_{1/2}$ denote edges in $E'$ that fail with probability 1/2~(Lemma~\ref{lm:chaingraph}). Then, step 2 takes $G'$ to output $F_{\mathcal{K}}$ such that $u_{G'}(P_{1/2}) =|R_{F_{\mathcal{K}}}|/2^M$ (Lemma~\ref{lm:pathcount}). Finally, $\unrelf =|R_{F_{\mathcal{K}}}|/2^M$.
\end{proof}

\begin{algorithm}
	\caption{{\relnet}}
	\label{alg:relnet}
	\begin{algorithmic}[1]
		\Statex \textbf{Input:} Instance {\instance} and $(\epsilon, \delta)$-parameters.
		\Statex \textbf{Output:} PAC estimate {\unrele}.
		\State Construct $G'$$=$$(V',E',\mathcal{K})$ replacing every edge $e\in E$ by $G^e$ such that $1-p_e=0.b_1\cdots b_{m_e}$ and $u_{G^e}(P_{1/2})=p_e$~(Lemma~\ref{lm:chaingraph}).
		\State Let $M = \sum_{e \in E} m_e = |E'|$, and construct $F_{\mathcal{K}}$ using $G'$ from Eq.~\ref{eq-relnet}.
		\State Invoke ApproxMC2, a hashing-based counting technique~\cite[][Chapter 4]{Meel2017}, to compute $\overline{\#F_{\mathcal{K}}}$, an approximation of $\#F_{\mathcal{K}}$ with $(\epsilon,\delta)$ guarantees.
		\Statex ${\unrelef} \gets \overline{\#F_{\mathcal{K}}}/2^{|M|}$
	\end{algorithmic}
\end{algorithm}
Steps 1-2 run in polynomial time on the size of $(G,P)$. Step 3 invokes ApproxMC2~\cite{Meel2017} to approximate $\#F_{\mathcal{K}}$. In turn, ApproxMC2 has access to a SAT-oracle, running in polynomial time on $\log 1/\delta$, $1/\epsilon$, and $|F_{\mathcal{K}}|$. Thus, relative to a SAT-oracle, {$\mathcal{K}$}-RelNet approximates {\unrel} with $(\epsilon,\delta)$ guarantees in the FPRAS theoretical sense. Also, we note that ApproxMC2's $(\epsilon, \delta)$ guarantees are for the multiplicative error $\Pr(1/(1+\epsilon)\unrelf\leq\unrelef\leq(1+\epsilon)\unrelf )\geq 1-\delta$~\cite{Meel2017}. This is a tighter error constraint than the relative error of Eq.~(\ref{eq:pac}), as one can show that $1-\epsilon \leq 1/(1+\epsilon)$ for $\epsilon\in(0,1)$. Thus, if an approximation method satisfies the multiplicative error guarantees, then it also satisfies the relative error guarantees. The converse is not true, and herein we will omit this advantage of {\relnet} over other methods for ease of comparison. Moreover, a SAT-oracle is a SAT-solver able to answer satisfiability queries with up to a million variables in practice. $F_{\mathcal{K}}$ has $|V'|+|E'|$ variables. {\relnet}'s theoretical guarantees now demand context relative to other existing methods, to then perform  computational experiments verifying its performance in practice.

\section{Context Relative to Competitive Methods}\label{sec:related_work}
This section briefly contextualizes our work relative to competitive techniques for network reliability evaluation, so as to facilitate the comparative analyses in Section~4. We arrange methods into three groups: exact or bounds, guarantee-less simulation, and probably approximately correct (PAC).
\subsection{Group (i): Exact or Bounds}
\noindent Network reliability belongs to the computational complexity class \#P-complete, which is largely believed to be intractable. This means that the task of computing {\unrel} efficiently is seemingly hopeless. While of limited application, the most popular techniques in this group employ approaches such as state enumeration~\cite{Ball1995}, direct decomposition~\cite{Dotson1979}, factoring~\cite{Satyanarayana1983}, or compact data structures like binary-decision-diagrams (BDD)~\cite{Hardy2007}. We refer the reader to the cited literature for a survey of exact methods~\cite{Ball1995,Le2014}.

The intractability of reliability problems motivates exploiting properties from graph theory. For example, in the case of bounded therewidth and degree, there are efficient algorithms available~\cite{Hardy2007, Canale2016}. Another promising family of methods issues fast converging bounds~\cite{Dotson1979,Le2013}, an approach that demonstrates practical performance even in earthquake engineering applications~\cite{Lim2012}, and that is applicable beyond connectivity-based reliability as part of the more general state-space-partition principle~\cite{Paredes2018,Alexopoulos1995}.
\subsection{Group (ii): Guarantee-less simulation}
\noindent When exact methods fail, guarantee-less simulations have found wide applicability. In the context of unbiased estimators,\footnote{The quality of a guarantee-less method being unbiased is key, as boosting confidence by means of repeating experiments leveraging the central limit theorem would lack justification otherwise.} a key property is the relative variance $\sigma^2_Y/\mu_Y^2$, with $Y$ a randomized Monte Carlo procedure such that $E[Y]={\unrelf}$. From Theorem~\ref{theom:markov} (Appendix), we know that should a method verify the bounded relative variance (BRV) property, i.e., $\sigma^2_Y/\mu_Y^2\leq C$ for $C$ some constant, then an efficient $(\epsilon,\delta)$ approximation is guaranteed with a sample size of $N=O(\epsilon^{-2}\log1/\delta)$. While certain methods verify the BRV property, the value of $C$ is typically unknown for general instances of the {\kterminal} problem, and thus the central limit theorem is often invoked for drawing confidence intervals despite known caveats~\cite{Bayer2014}. Some techniques verifying the BRV property include the permutation Monte Carlo-based Lomonosov's Turnip (LT)~\cite{Gertsbakh2016} and its sequential splitting extension, the Split-Turnip (ST)~\cite{Vaisman2016}, and the importance sampling variants of the recursive variance reduction (RVR) algorithm~\cite{Cancela2014}. They significantly outperform the crude Monte Carlo (CMC) method in the rare event setting, with RVR even displaying the VRV property in select instances, as evidenced in empirical evaluations.

As we noted, the number of samples in the crude Monte Carlo approach scales like $1/{\unrelf}$, which can be problematic in highly-reliable systems. A more promising approach leverages the Markov Chain Monte Carlo method and the product estimator~\cite{Jerrum1988,Fishman1994}, where the small {\unrel} estimation is bypassed by estimating the product of larger quantities. Significantly, the sample size roughly scales like $\log1/{\unrelf}$~\cite{Dyer1991}. The product estimator is popularly referred to as multilevel splitting as it has independently appeared in other disciplines~\cite{Glasserman1999,kahn1951estimation,rosenbluth1955monte}, and even more recently in the civil and mechanical engineering fields under the name of subset simulation~\cite{Au2001}. In the case of network reliability, the latent variable formulation by Botev et al.~\cite{Botev2012}, termed generalized splitting (GS),  delivers unbiased estimates of {\unrel}. The similar approach by Zuev et al.~\cite{Zuev2015} is not readily applicable to the {\kterminal} and delivers biased samples, which can be an issue to rigorously assess confidence.
\subsection{Group (iii): PAC methods}
\noindent In a breakthrough paper, Karger gave the first efficient approximation for the all-terminal network unreliability problem~\cite{Karger2001}. However, Karger's algorithm is not always practical despite recent improvement~\cite{Karger2016}. Also, unlike {\relnet}, Karger's algorithm is not readily applicable to the more general $\mathcal{K}$-terminal network reliability problem. 

Besides our network reliability PAC approximation technique,  {\relnet}, and that is specialized to the {\kterminal} problem, there are general Monte Carlo sampling schemes that deliver $(\epsilon,\delta)$ guarantees. The reminder of this subsection highlights relevant methods that are readily implementable in Monte Carlo-based network reliability calculations. 

Denoting $Y$ the random samples produced by unbiased sampling-based estimators, traditional simulation approaches take the average of i.i.d. samples of $Y$. Such estimators can be integrated into optimal Monte Carlo simulation (OMCS) algorithms~\cite{Dagum2000}. An algorithm $A$ is said to be optimal (up to a constant factor) when its sample size $N_A$ is not proportionally larger in expectation than the sample size $N_B$ of any other algorithm $B$ that is also an $(\epsilon,\delta)$ randomized approximation of $\mu_Y$, and that has access to the same information as $A$, i.e., $E[N_A] \leq c \cdot E[N_B]$ with $c$ a universal constant.

\begin{algorithm}
	\caption{ Stopping Rule Algorithm ({\sra})~\cite{Dagum2000}.}
	\small
	\label{alg:sra}
	\begin{algorithmic}[1]
		\Statex \textbf{Input:} $\epsilon, \delta\in(0,1)$ and random variable $Y$.
		\Statex \textbf{Output:} Estimate $\unrelef$ with PAC guarantees.
		\Statex Let $\{Y_i\}$ be a set of i.i.d samples of $Y$.
		\Statex Compute constants $\Upsilon = 4(e-2)\log(2/\delta)1/\epsilon^2$, $\Upsilon_1 = 1+(1+\epsilon)\cdot\Upsilon$.
		\Statex Initialize $S\gets 0$, $N\gets 0$.
		\Statex \textbf{while} $(S<\Upsilon_1)$ \textbf{do:} $N\gets N + 1$, $S\gets S+Y_N$.
		\Statex $\unrelef\gets \Upsilon_1/N$
	\end{algorithmic}
\end{algorithm}
A simple and general purpose black box algorithm to approximate {\unrel} with PAC guarantees is the \textit{Stopping Rule Algorithm} ({\sra}) introduced by Dagum et al.~\cite{Dagum2000}. The convergence properties of {\sra} were shown through the theory of martingales and its implementation is straightforward (Algorithm~\ref{alg:sra}).

Even though {\sra} is optimal up to a constant factor for RVs with support $\{0,1\}$, a different algorithm and analysis leads to the \textit{Gamma Bernoulli Approximation Scheme} ({\gbas})~\cite{Huber2017}, which improves the expected sample size by a constant factor over {\sra} and demonstrates superior performance in practice due to improved lower order terms in its guarantees. {\gbas} has the additional advantage with respect to {\sra} of being unbiased, and it is relatively simple to implement. The core idea of {\gbas} is to construct a RV such that its relative error probability distribution is known. The procedure is shown in Algorithm~\ref{alg:gbas}, $I$ is the indicator function, $\text{Unif}(0,1)$ is a random draw from the uniform distribution bounded in $[0,1]$, and $\text{Exp}(1)$ is a random draw from an exponential distribution with parameter $\lambda=1$. Also, Algorithm~\ref{alg:gbas} requires parameter $k$, which is set as the smallest value that guarantees $\delta \geq \Pr(\mu_Y/\hat\mu_Y<(1+\epsilon)^2 \text{ or } \mu_Y/\hat\mu_Y>(1-\epsilon)^2 )$ with $\mu_Y/\hat\mu_Y\sim \text{Gamma}(k, k-1)$~\cite{Huber2017}. In practice, values of $k$ for relevant $(\epsilon,\delta)$ pairs can be tabulated. Alternatively, if one can evaluate the cumulative density function (cdf) of a Gamma distribution, galloping search can be used to find the optimal value of $k$ with logarithmic overhead (on the number of cdf evaluations).

\begin{algorithm}
	\caption{ Gamma Bernoulli Approximation Scheme ({\gbas})~\cite{Huber2017}.}
	\small
	\label{alg:gbas}
	\begin{algorithmic}
		\Statex \textbf{Input:} $k$ parameter.
		\Statex \textbf{Output:} Estimate $\unrelef$ with PAC guarantees.
		\Statex Let $\{Y_i\}$ be a set of independent samples.
		\Statex Initialize $S \gets 0$, $R\gets 0$, $N\gets 0$.
		\While{$(S \neq k)$}
		\State $N\gets N + 1$, $B \gets I(\text{Unif}(0,1) \leq Y_N )$
		\State $S\gets S+ B$, $R \gets R + \text{Exp}(1)$
		\EndWhile
		\Statex $\unrelef\gets \Upsilon_1/N$
	\end{algorithmic}
\end{algorithm}
Note that {\sra} and {\gbas} give PAC estimates with optimal expected number of samples for RVs with support $\{0,1\}$, yet they disregard the variance reduction properties of more advanced techniques. Thus, one can ponder, is there a way to exploit a randomized procedure $Y$ such that $\sigma_{Y} \ll \sigma_{Y^{CMC}}$ in the context of OMCS? The \textit{Approximation Algorithm} (\aaa), introduced by Dagum et al.~\cite{Dagum2000}, and based on sequential analysis~\cite{wald1973sequential}, gives a partially favorable answer. In particular, steps 1 and 2 (Algorithm~\ref{alg:aa}) are trial experiments that give rough estimates of $\mu_Y$ and $\sigma_Y^2/\mu_Y^2$, respectively. Then, step 3 is the actual experiment that outputs {\unrele} with PAC guarantees. {\aaa} assumes $Y\in[0,1]$, and it was shown to be optimal up to a constant factor.

The downside of {\aaa}, or any OMCS algorithm as {\aaa} is optimal, is that it requires in expectation $N_{\aaaf}=O(\max\{\sigma^2_{Y}/\mu_Y^2,\epsilon/\mu_Y\}\cdot\epsilon^{-2}\ln1/\delta)$ samples. Thus, despite considering the relative variance $\sigma^2_{Y}/\mu_Y^2$, OMCS algorithms become impractical in the rare-event regime. For example, consider the case in which edge failure probabilities tend to zero and $1/\mu_Y$ goes to infinity. If a technique delivers $Y$ that meets the BRV property, i.e., $\sigma^2_{Y}/\mu_Y^2\leq C$ for $C$ some constant, then, from Theorem~\ref{theom:markov} (Appendix), we know a sample of $N=O(\epsilon^{-2}\ln1/\delta)$ suffices, meanwhile $N_{\aaaf}\rightarrow\infty$.

\begin{algorithm}
	\caption{The Approximation Algorithm ({\aaa})~\cite{Dagum2000}.}
	\label{alg:aa}
	\begin{algorithmic}[1]
		\Statex \textbf{Input:} $(\epsilon, \delta)$-parameters.
		\Statex \textbf{Output:} Estimate $\unrelef$ with PAC guarantees.
		\Statex Let $\{Y_i\}$ and $\{Y'_i\}$ be two sets of independent samples of $Y$.
		\State $\epsilon'\gets \min\{1/2,\sqrt{\epsilon}\}$, $\delta'\gets\delta/3$
		\Statex $\hat{\mu}_Y\gets \textit{SRA}(\epsilon',\delta')$
		\State $\Upsilon\gets4(e-2)\log(2/\delta)1/\epsilon^2$
		\Statex $\Upsilon_2\gets 2(1+\sqrt{\epsilon})(1+2\sqrt{\epsilon})(1+\ln(3/2)/\ln(2/\delta))\Upsilon$
		\Statex $N\gets\Upsilon_2\cdot\epsilon/\hat\mu_Y$, $S\gets0$
		\Statex \textbf{for} ($i=1,\dots,N$) \textbf{do}: $S\gets S+(Y'_{2i-1}-Y'_{2i})^2/2$
		\Statex $\hat{r}^*_Y \gets \max\{S/N,\epsilon\hat{\mu_Y}\}/\hat\mu_Y^2$
		\State $N_{\aaaf}\gets \Upsilon_2 \cdot\hat{r}^*_Y$
		\Statex \textbf{for} ($i=1,\dots,N_{\aaaf}$) \textbf{do}: $S\gets S+Y_i$
		\Statex $\unrelef\gets S/N$
	\end{algorithmic}
\end{algorithm}
We will use {\gbas} for CMC with $(\epsilon,\delta)$ guarantees, and use {\aaa}, given its generality, to turn various existing techniques into PAC methods themselves. For {\aaa}, note that the rough estimate $\hat{\mu}_Y$ in step 1 is computed using $Y^{CMC}$ as it is the cheapest, but from step 2 and on, the estimator that is intended to be tested is used, but the reported runtime will be that of step 3 to measure variance reduction and runtime without trial experiments.

\section{Computational Experiments}\label{sec:experiments}
\noindent
A fair way to compare methods is to test them against challenging benchmarks and quantify empirical measures of performance relative to their theoretical promise. We take this approach to test $\mathcal{K}$-RelNet alongside competitive methods. The following subsections describe our experimental setting, listing implemented methods, and their application to various benchmarks.

\subsection{Implemented estimation methods}

Table~\ref{t:methods} lists reliability evaluation methods that we consider in our numerical experiments. Exact methods run until giving an exact estimate or best bounds until timeout. Each guarantee-less simulation method uses a custom number of samples $N$ that depends on the shared parameter $N_S$ (Table~\ref{t:methods}). This practice, borrowed from Vaisman et al.~\cite{Vaisman2016}, tries to account for the varying computational cost of samples among methods.

PAC algorithms {\aaa} or {\gbas} are use in combination with guarantee-less sampling methods to compare runtime given a target precision. For example, {\gbas}($Y^{\text{CMC}}$) denotes Algorithm~\ref{alg:gbas} when samples are drawn from the CMC estimator. Experiments with {\aaa} use $(\epsilon,\delta)=(0.2,0.2)$. Experiments embedded in {\gbas} use two configurations:  $(0.2,0.2)$ and $(0.2,0.05)$. {\relnet} uses $(0.8,0.2)$ to avoid time outs. As we will verify, in practice, PAC-methods issue estimates with better precision than the input theoretical $(\epsilon,\delta)$-guarantees.
\begin{table}[]
	\centering
	\caption{Methods used in computational experiments, and corresponding parameters.}
	\label{t:methods}
	{\scriptsize
		\begin{tabular}{|c|l|l|l|l|c|}
			\hline
			\multicolumn{1}{|c|}{\textbf{Group}} &
			\multicolumn{1}{c|}{\textbf{Methods}}             & \multicolumn{1}{c|}{\textbf{IDs}} & \multicolumn{1}{c|}{\textbf{Parameters}} & \multicolumn{1}{c|}{\textbf{Ref.}} \\ \hline
			i & BDD-based Network Reliability                           & HLL   &  n/a   &  \cite{Hardy2007,Herrmann2010}  \\ \hline
			\multirow{4}{*}{ii} & Lomonosov's-Turnip                                & LT            &    $N=N_S$               &                               \cite{Gertsbakh2016}          \\\cline{2-5}
			& Sequential Splitting Monte Carlo                  & ST           &      $B=100$, $N=N_S/B$             &                               \cite{Vaisman2016}           \\\cline{2-5}
			&Generalized Splitting                             & GS              &     $s=2, N_0 =10^3, N=N_S$            &                          \cite{Botev2012}        \\\cline{2-5}
			&Recursive Variance Reduction & RVR      &      $N=N_S/\binom{|\mathcal{K}|}{2}$            &                              \cite{Cancela2014}           \\\hline
			\multirow{3}{*}{iii} & Karger's 2-step Algorithm          & K2Simple              &  $\epsilon,\delta$          &                              \cite{Karger2016}           \\ \cline{2-5}
			iii & Optimal Monte Carlo Simulation                                 & 
			{\gbas, \aaa}    &        $\epsilon,\delta$                  &\cite{Dagum2000, Huber2017}  \\\cline{2-5}
			& Counting-based Network Unreliability              & {\relnet}     &   $\epsilon,\delta$      &                                   This paper  \\ \hline
		\end{tabular}
	}
\end{table}

To the best of our knowledge, methods in Table~\ref{t:methods} are some of the best in their categories as evidenced in the literature. We implemented all methods in a Python prototype for uniform comparability and ran all experiments in the same machine---a 3.60GHz quad-core Intel i7-4790 processor with 32GB of main memory and each experiment was run on a single core.

\subsection{Estimator Performance Measures}
We use the next empirical measures to assess the performance of reliability estimation methods. Let $\hat{u}$ be an approximation of $u$. We measure the \textit{observed multiplicative error} {\epso} as $(\hat{u}-u)/u$ if $\hat{u}>u$, and $(\hat{u}-u)/\hat{u}$ otherwise. Also, for a fixed PAC-method, target relative error $\epsilon$, and independent measures $\epsof^{(1)},\dots,\epsof^{(M)}$, we compute the \textit{observed confidence} parameter {\deltao} as $1/M\cdot\sum_{i=1}^{M}\mathbbm{1}(|\epsof^{(i)}|\geq\epsilon)$. Satisfaction of ($\epsilon$, $\delta$) is guaranteed, but {\epso} and {\deltao} can expose theoretical guarantees that are too conservative. 

Furthermore, for guarantee-less sampling methods we measure {\epso} but not {\deltao}, as these do not support confidence a priori. Thus, we use empirical measures of variance reduction to assess the desirability of sampling techniques over the canonical method (CMC). Let $\sigma^2_{Y^{CMC}}=\mu_Y\cdot(1-\mu_{Y})/N$ be the variance associated to CMC, and let $\sigma_{Y^A}^2$ be the sample associated to method $A$. Clearly, $\sigma^2_{Y^{CMC}}/\sigma^2_{Y^{A}}>1$ will favor $A$ over CMC. However, this is not the only important consideration in practice. For respective CPU times $\tau_{Y^\text{CMC}}$ and $\tau_{Y^A}$, a ratio $\tau_{Y^{CMC}}/\tau_{Y^{A}}<1$ would imply a higher computational cost for $A$. To account for both, variance and CPU time, we use the \textit{efficiency ratio}, defined as $\erf(Y^A) = \big(\sigma^2_{Y^{CMC}}/\sigma^2_{Y^{A}})\cdot \big(\tau_{Y^{CMC}}/\tau_{Y^{A}})$~\cite{Fishman1986c}. In practice, when $\erf(Y^A)<1$, one prefers the more straightforward CMC. A similar measure in the literature is the \textit{work normalized relative variance}~\cite{Botev2012}, defined as $\text{wnrv}(Y)=\tau_{Y}\sigma^2_Y/\mu_Y^2$, which is related to the efficiency ratio via $\erf(Y^A)=\text{wnrv}(Y^\textit{CMC})/\text{wnrv}(Y^A)$. We prefer $\erf(Y^A)$ over $\text{wnvr}(Y^A)$ as it is, ipso facto, a measure of adequacy of $A$ over CMC, informing users on whether they need to implement a more sophisticated method than CMC.\footnote{The ratio $\sigma^2_{Y^{\textit{CMC}}}/\sigma^2_{Y^{A}}$ in the {$\erf$} is also the ratio of the relative variances of $Y^{\textit{CMC}}$ and $Y^{A}$, shedding light on how many times larger (or smaller) the sample associated to CMC needs to be with respect to $A$ from Theorem~\ref{theom:markov} (Appendix).} 

The next subsections introduce the benchmarks we use and discussion of results. Also, in our benchmarks we consider sparse networks, i.e. $|E|=O(|V|)$, which resemble engineered systems.

\subsection{Rectangular Grid Networks}
We consider $N$$\times$$N$ square grids (Figure~\ref{fg:gridgraph}) because they are irreducible (via series-parallel reductions) for $N>2$, their tree-width is exactly $N$, and they can be grown arbitrarily large until exact methods fail to return an estimate. Also, failure probabilities can be varied to challenge simulation methods. Our goal is to increase $N$ and vary failure probabilities uniformly to verify running time, scalability, and quality of approximation. We evaluate performance until methods fail to give a desirable answer. In particular, we consider values of $N$ in the range $2$ to $100$. Also, assume all edges fail with probability $2^{-i}$, for $i\in\{1,3,\dots,15\}$. Furthermore, we consider extreme cases of $\mathcal{K}$~(Figure~\ref{fg:gridgraph}), namely, all-terminal and two-terminal reliability, and a $\mathcal{K}$-terminal case with terminal nodes distributed in a checkerboard pattern.

\begin{figure}
	\centering
	\pgfmathtruncatemacro{\n}{3}
	\pgfmathtruncatemacro{\ni}{\n-1}
	\pgfmathtruncatemacro{\nsize}{1}
	\pgfmathtruncatemacro{\uselabel}{0}
	\pgfmathsetmacro{\mult}{0.7}
	\ifthenelse{\isodd{\ni}}{\pgfmathtruncatemacro{\aux}{1}}{\pgfmathtruncatemacro{\aux}{0}}
	\footnotesize
	\begingroup
	\captionsetup[subfigure]{width=1.1in}
	\subfloat[All-Terminal]
	{
		\begin{tikzpicture}[darkstyle/.style={circle,draw,fill=gray!65,minimum size=\nsize},clearstyle/.style={circle,draw,fill=gray!10,minimum size=\nsize}]
		\foreach \x in {0,...,\n}
		\foreach \y in {0,...,\n}
		{\pgfmathtruncatemacro{\label}{\x + (\n+1) *  (\y) + 1}
			\node [darkstyle]  (\x\y) at (\mult*\x,\mult*\y) {\ifthenelse{\equal{\uselabel}{1}}{\label}{}};}
		\foreach \x in {0,...,\n}
		\foreach \y [count=\yi] in {0,...,\ni}
		\draw (\x\y)--(\x\yi) (\y\x)--(\yi\x);
		\end{tikzpicture}
	}
\qquad\qquad
	\subfloat[Two-Terminal]
	{
		\begin{tikzpicture}[darkstyle/.style={circle,draw,fill=gray!65,minimum size=\nsize},clearstyle/.style={circle,draw,fill=gray!10,minimum size=\nsize}]
		\foreach \x in {0,...,\n} {
			\foreach \y in {0,...,\n} {
				\pgfmathtruncatemacro{\test}{\x + (\n+1) *  (\y) + 1 }
				\pgfmathtruncatemacro{\nn}{(\n+1)^2}
				\pgfmathtruncatemacro{\label}{\x + (\n+1) *  (\y) + 1}
				\ifthenelse {\equal{1}{\test} \OR \equal{\nn}{\test}}
				{\node [darkstyle]  (\x\y) at (\mult*\x,\mult*\y) {\ifthenelse{\equal{\uselabel}{1}}{\label}{}};}
				{\node [clearstyle]  (\x\y) at (\mult*\x,\mult*\y) {\ifthenelse{\equal{\uselabel}{1}}{\label}{}};}
			}
		}
		\foreach \x in {0,...,\n}
		\foreach \y [count=\yi] in {0,...,\ni}
		\draw (\x\y)--(\x\yi) (\y\x)--(\yi\x);
		\end{tikzpicture}    
	}
\qquad\qquad
	\subfloat[$\mathcal{K}$-Terminal\label{sfig:check}]
	{
		\begin{tikzpicture}[darkstyle/.style={circle,draw,fill=gray!65,minimum size=\nsize},clearstyle/.style={circle,draw,fill=gray!10,minimum size=\nsize}]
		\foreach \x in {0,...,\n} {
			\foreach \y in {0,...,\n} {
				\pgfmathtruncatemacro{\test}{\x + (\n+\aux) *  (\y) + 1 }
				\pgfmathtruncatemacro{\label}{\x + (\n+1) *  (\y) + 1}
				\ifthenelse {\isodd{\test}}
				{\node [darkstyle]  (\x\y) at (\mult*\x,\mult*\y) {\ifthenelse{\equal{\uselabel}{1}}{\label}{}};}
				{\node [clearstyle]  (\x\y) at (\mult*\x,\mult*\y) {\ifthenelse{\equal{\uselabel}{1}}{\label}{}};}
			}
		}
		\foreach \x in {0,...,\n}
		\foreach \y [count=\yi] in {0,...,\ni}
		\draw (\x\y)--(\x\yi) (\y\x)--(\yi\x);
		\end{tikzpicture}
	}
	\endgroup
	\pgfmathparse{int(\n+1)}
	\caption{Example of an $N\times N$ grid graph, with $N = 4$. Darkened nodes belong to the terminal set $\mathcal{K}$.}
	\label{fg:gridgraph}
\end{figure}
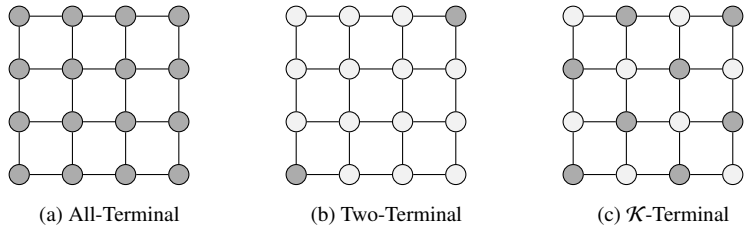

\subsubsection{Exact calculations}\label{sec:exact}
\noindent For reference, we obtained exact unreliability calculations using the BDD-based method by Hardy, Lucet, and Limnios~\cite{Hardy2007}, herein termed {\hll} due to its authors. We computed {\unrel} for $N=2,..,10$ and all values of $p_e$. Figure~\ref{fg:bdd} shows a subset of exact estimates (a-b) and exponential scaling of running time (c). Several other exact methods we studied and referenced in Section~3, were used, but {\hll} was the only one that managed to estimate {\unrel} exactly for all $N\leq10$. However, {\hll} became memory-wise more consuming for $N>10$. Thus, if memory is the only concern, the state-space partition can be used instead to get anytime bounds on {\unrel} at the expense of larger runtime, but storing at most $O(|E|)$ vectors $X\in\{0,1\}^m$ simultaneously~\cite{Paredes2018}. Next, we use these exact estimates to compute {\epso} and {\erf} for guarantee-less simulation methods, and to compute {\epso} and {\deltao} for PAC methods.
\begin{figure}
	\centering
	\subfloat[{\allterminal}.]{\includegraphics[width=0.33\textwidth]{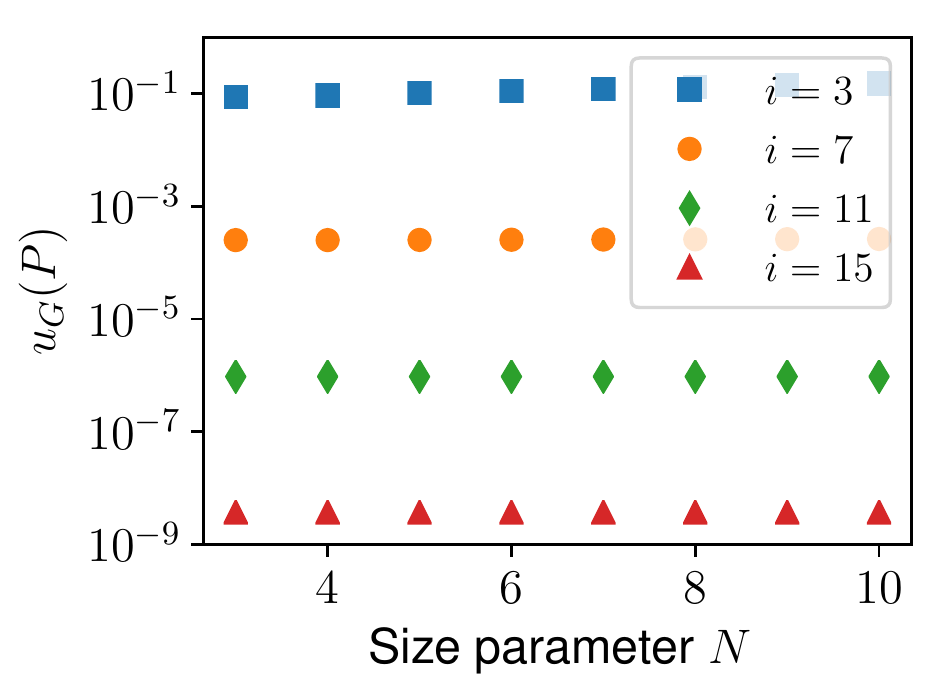}}\hfill
	\subfloat[{\twoterminal}.]{\includegraphics[width=0.33\textwidth]{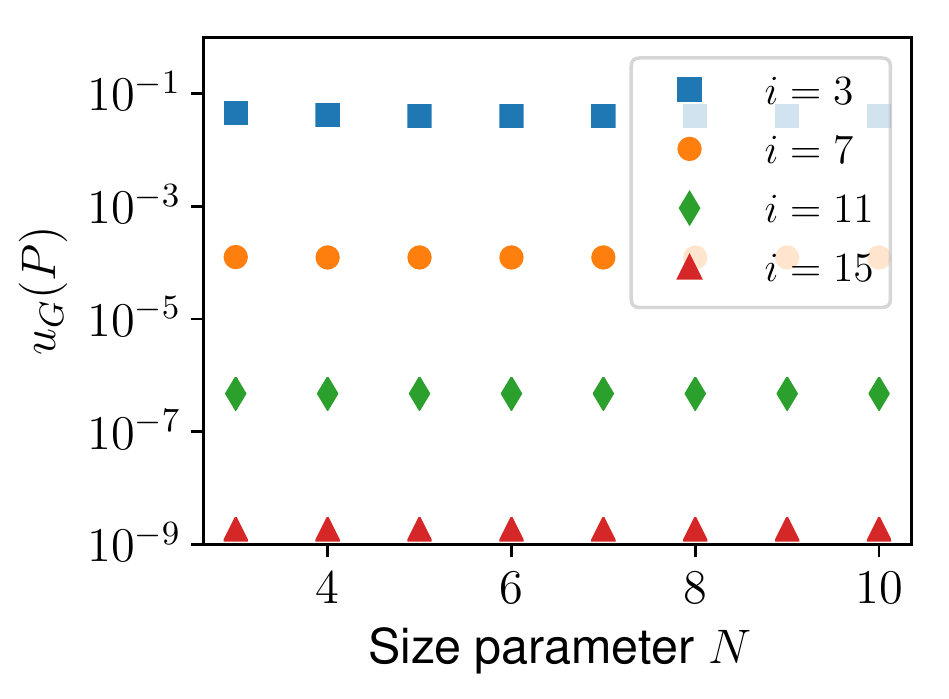}}\hfill
	\subfloat[{all cases}.]{\includegraphics[width=0.33\textwidth]{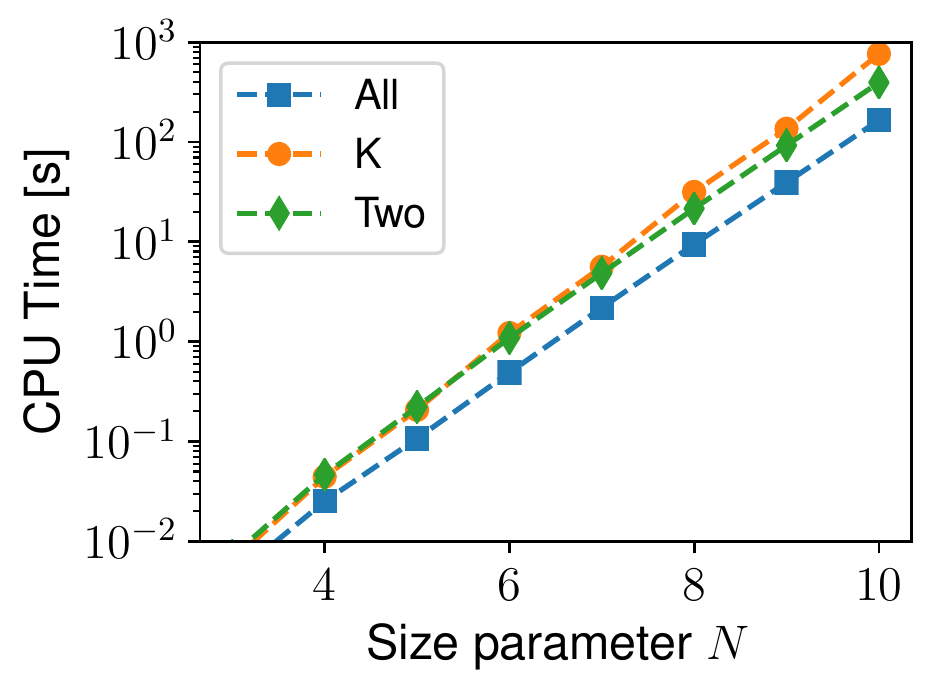}}\hfill
	\caption{(a-b) Exact estimates of {\unrel} and (c) CPU time using HLL for all-, two-, and {$\mathcal{K}$}-terminal cases.}
	\label{fg:bdd}
\end{figure}

\subsubsection{Guarantee-less simulation methods} \noindent Figure~\ref{fg:grids_eps} shows values of {\epso}  for the case of {\twoterminal} and setting $N_S=10^4$. Most values are below the $\epsof=0.2$ threshold. For RVR we observed values of {\epso} in the order of the float-point precision for the largest values of $i$. We attribute this to the small number of cuts with maximum probability (2-4 in our case) that, together with the fact that RVR finds them all in the decomposition process, endows RVR with the VRE property in this case. Conversely, other methods do not rely as heavily on these small number of cuts.
\begin{figure}
	\centering
	\subfloat[LT]{\includegraphics[width=0.25\textwidth]{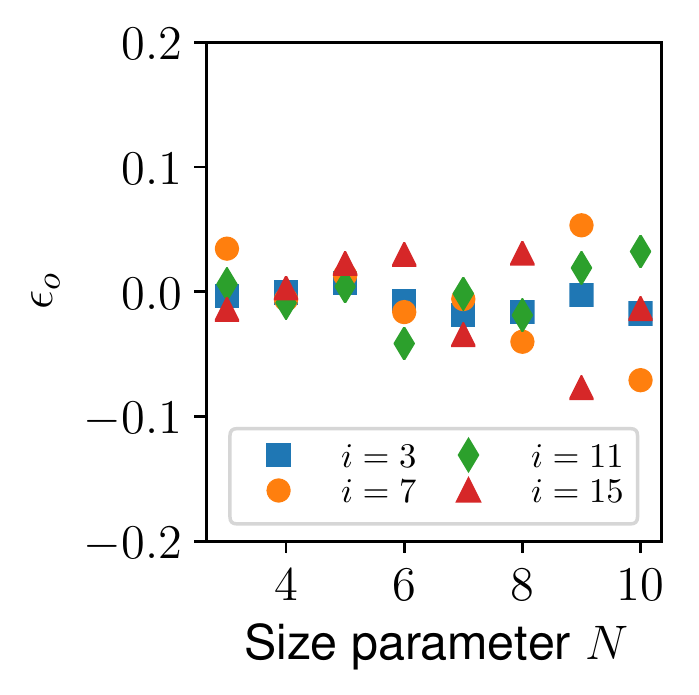}}
	\subfloat[ST]{\includegraphics[width=0.25\textwidth]{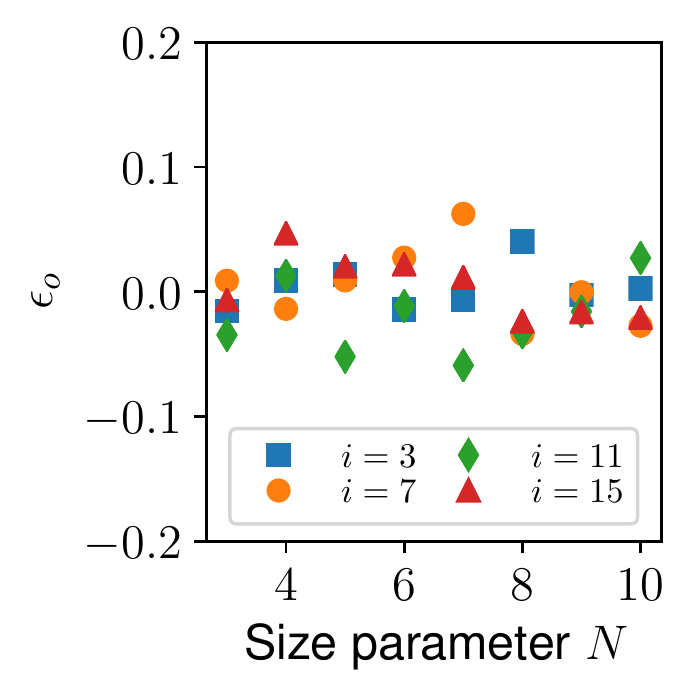}}
	\subfloat[GS]{\includegraphics[width=0.25\textwidth]{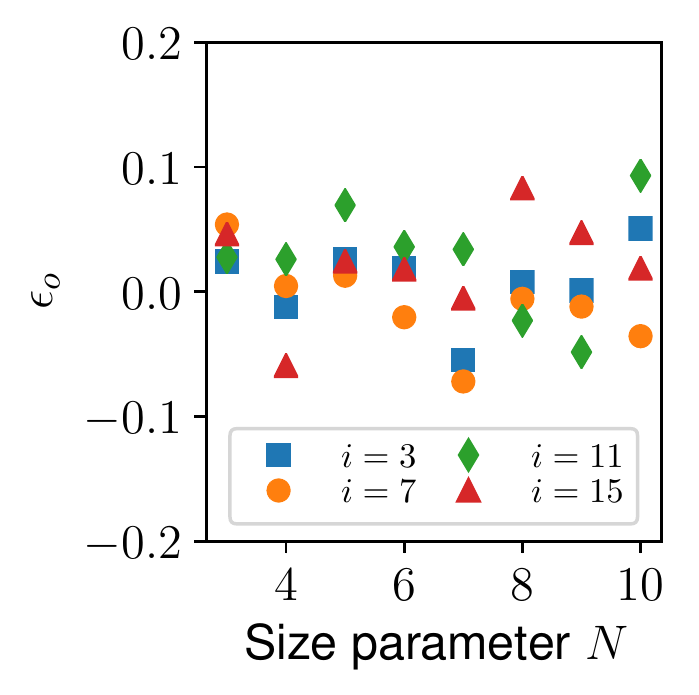}}
	\subfloat[RVR]{\includegraphics[width=0.25\textwidth]{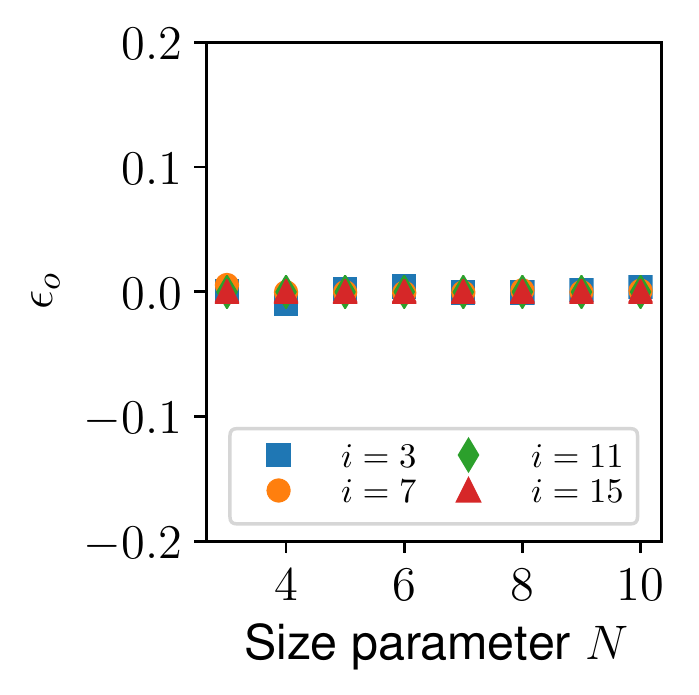}}
	\caption{Multiplicative error {\epso} for guarantee-less simulation methods in the {\twoterminal} case.}
	\label{fg:grids_eps}
\end{figure}

Moreover, the CPU time varied among methods as shown in Figure~\ref{fg:grids_time}. The only method whose single sample computation is affected by the values of $i$ is GS, consistent with the expected number of levels, which scales as $\log1/\unrelf$. However, matrix exponential operations for handling more cases of $i$ added overhead in LT and ST; the sharp time increase from $N=5$ to $N=6$ is due to this operation, consistent with findings by Botev et al.~\cite{Botev2012}. Instead, RVR does not suffer from numerical issues and appears to verify the VRV property in this grid topology.
\begin{figure}
	\centering
	\subfloat[LT]{\includegraphics[width=0.25\textwidth]{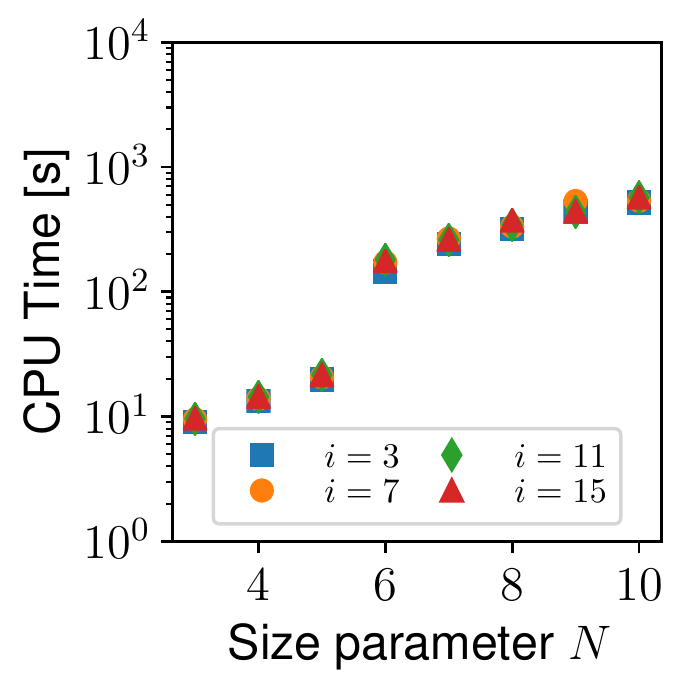}}
	\subfloat[ST]{\includegraphics[width=0.25\textwidth]{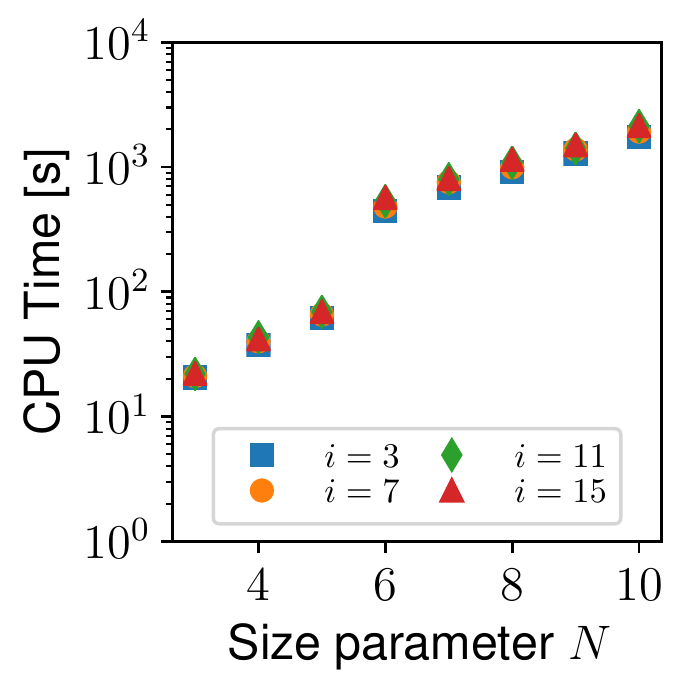}}
	\subfloat[GS]{\includegraphics[width=0.25\textwidth]{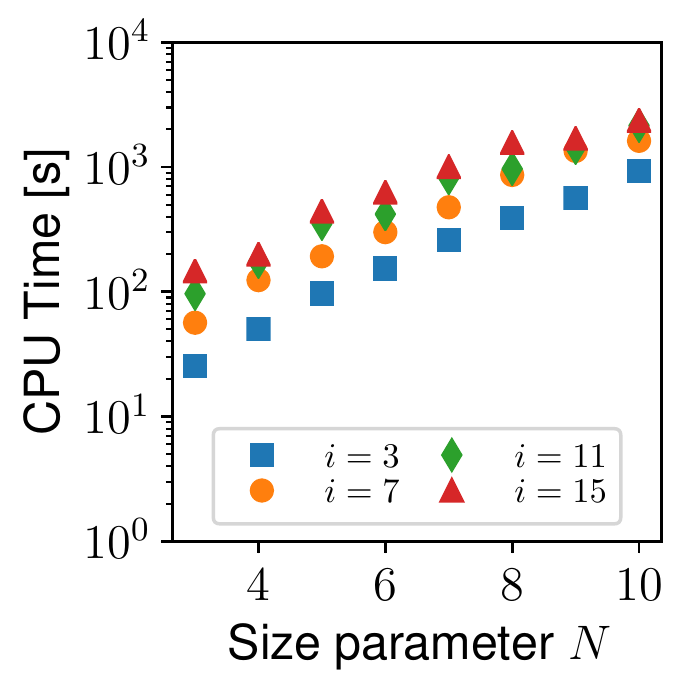}}
	\subfloat[RVR]{\includegraphics[width=0.25\textwidth]{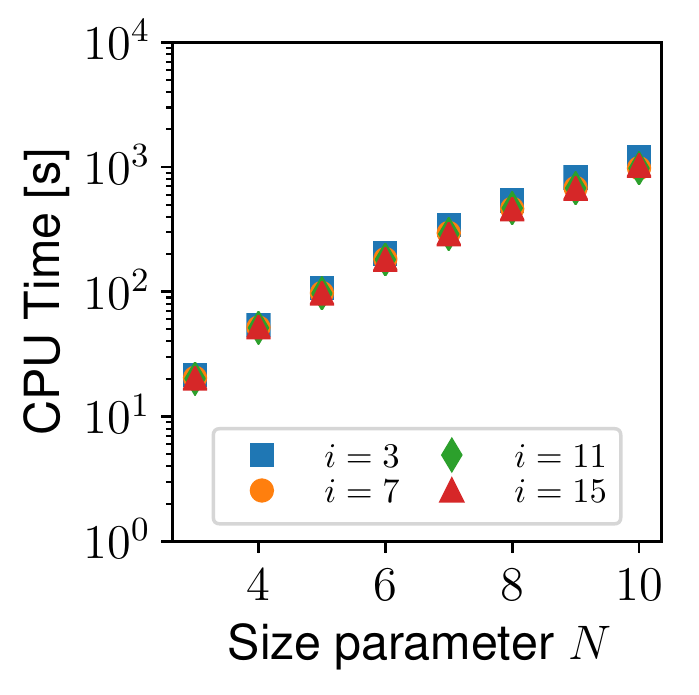}}
	\caption{CPU time for guarantee-less simulation methods in the {\twoterminal} case.}
	\label{fg:grids_time}
\end{figure}

Also, to compare all methods in a uniform fashion we used the efficiency ratio (Figure~\ref{fg:grids_er}). Values of $\sigma_{Y^{\text{CMC}}}^2$ for computing the efficiency ratio are exact from HLL, and CPU time $\tau_{Y^{\text{CMC}}}$ is based on $10^4$ samples. Estimates below the horizontal line are less reliable than those obtained with CMC for the same amount of CPU time. In particular, we note that for less rare failure probabilities ($2^{-7}\approx0.008$) some methods fail to improve over CMC. Missing values for RVR show improvements above $10^7$ in the efficiency ratio which, again, can be attributed to it meeting the VRE property in these benchmarks. Furthermore, an interesting trend among simulation methods is that there is a downward trend in their efficiency ratio as $N$ grows. Thus, we can construct an arbitrarily large squared grid for some $N$ that will, ceteris paribus, yield an efficiency ratio below 1 in favor of CMC. We attribute this to the time complexity of CMC samples in sparse graphs, which can be computed in $O(|V|)$ time, whereas other techniques run in $O(|V|^2)$ time or worse. Thus, the larger the graph the far greater the cost per sample by more advanced techniques with respect to CMC. 
\begin{figure}
	\centering
	\subfloat[LT]{\includegraphics[width=0.25\textwidth]{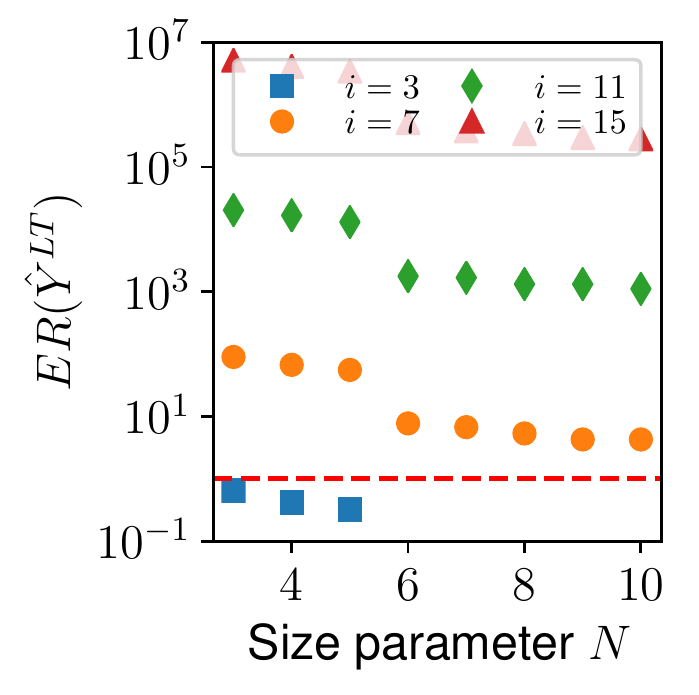}}
	\subfloat[ST]{\includegraphics[width=0.25\textwidth]{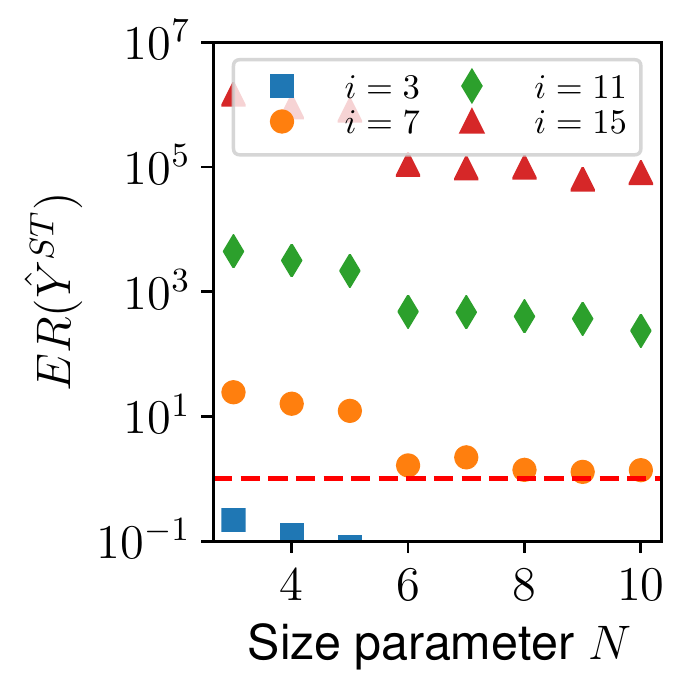}}
	\subfloat[GS]{\includegraphics[width=0.25\textwidth]{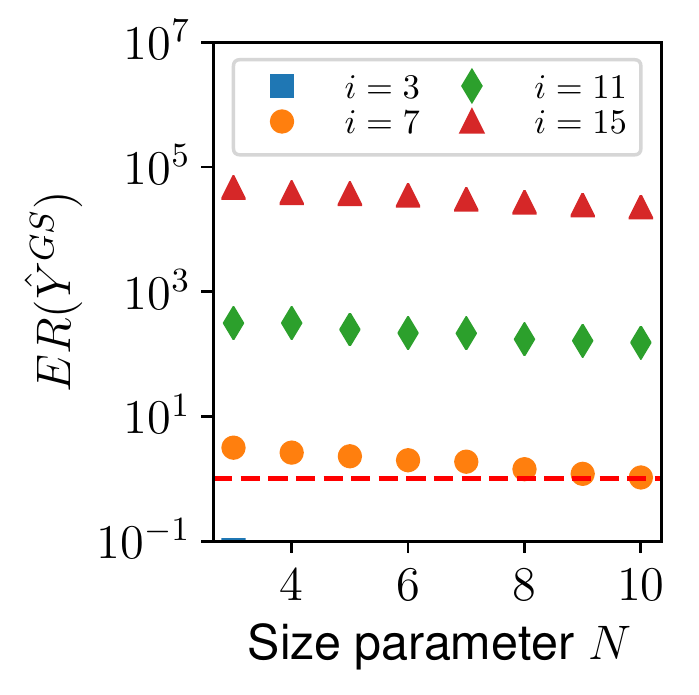}}
	\subfloat[RVR]{\includegraphics[width=0.25\textwidth]{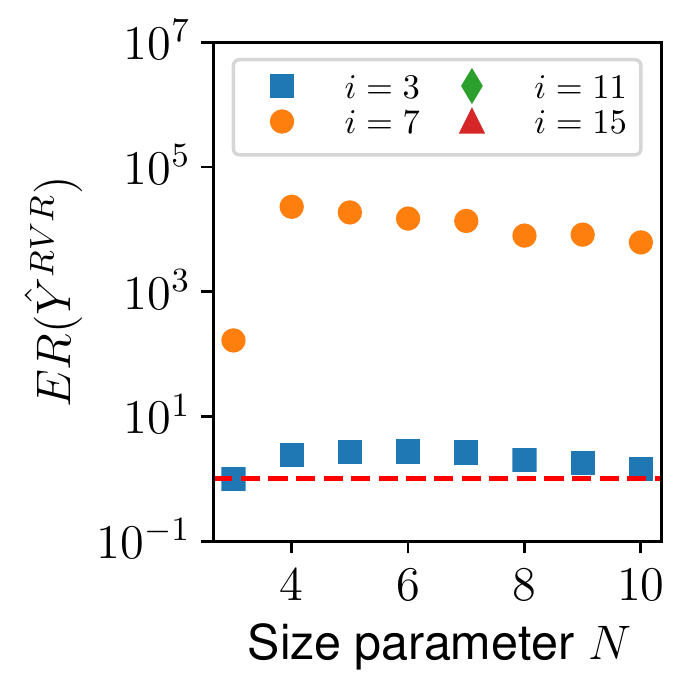}}
	\caption{{\erf} for guarantee-less simulation methods in the {\twoterminal} case.}
	\label{fg:grids_er}
\end{figure}

\subsubsection{Probably approximately correct (PAC) methods}

\noindent Next, we embedded simulation methods in {\aaa}, except CMC which was run using {\gbas} because the latter is optimal for Bernoulli RVs such as $Y^\textit{CMC}$. Figure~\ref{fg:grids_paced} shows the runtime for methods embedded into {\aaa}. We were able to feasible compute PAC-estimates for edge failure probabilities of $2^{-5}\approx0.03$ or larger. The approximation guarantees turned out to be rather conservative, obtaining far better precision in practice. Variance reduction through {\aaa} can only reduce sample size by a factor of $O(1/\epsilon)$ with respect to the Bernoulli case (i.e. $N_{ZO}$), thus PAC-estimates with advanced simulation methods using {\aaa} seem to be confined to cases where ${\unrelf\geq 0.005}$ for the square grids benchmarks. However, conditioned on disruptive events such as natural disasters in which failure probabilities are larger, {\aaa} can deliver practical PAC-estimates.
\begin{figure}
	\centering
	\subfloat[LT]{\includegraphics[width=0.25\textwidth]{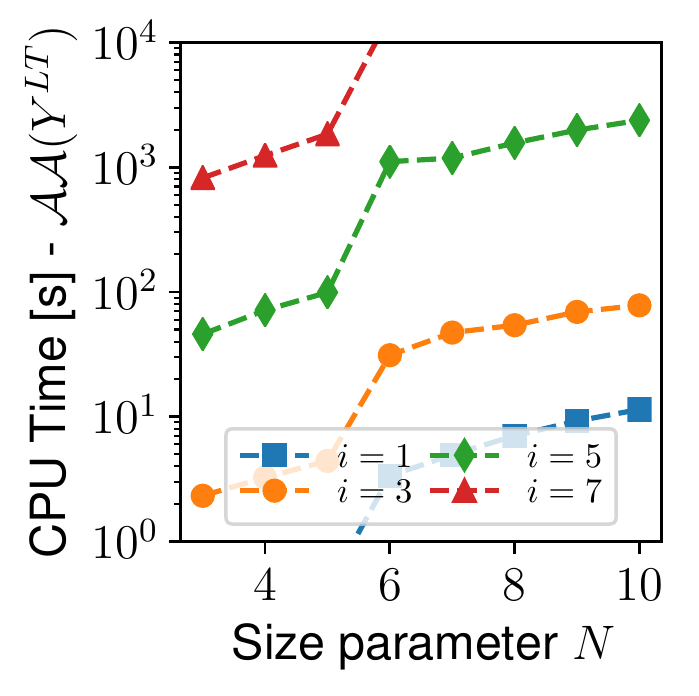}}
	\subfloat[ST]{\includegraphics[width=0.25\textwidth]{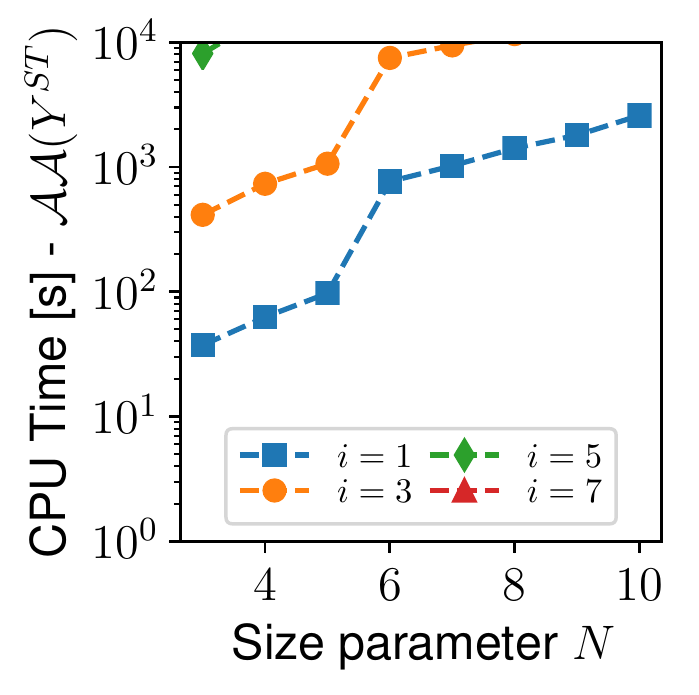}}
	\subfloat[GS]{\includegraphics[width=0.25\textwidth]{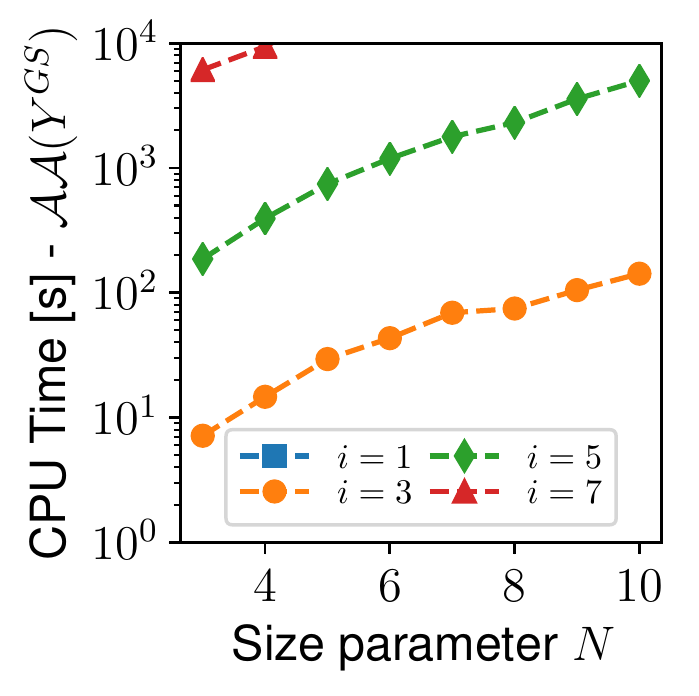}}
	\subfloat[RVR]{\includegraphics[width=0.25\textwidth]{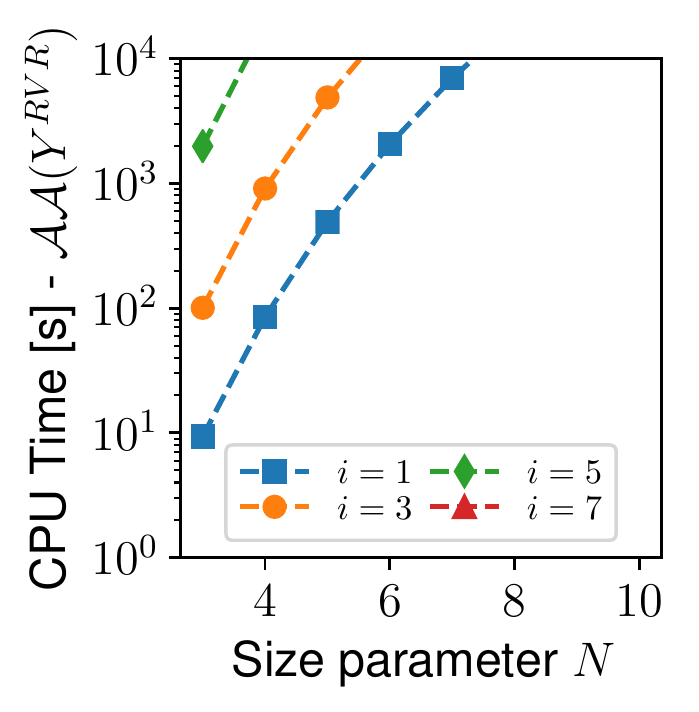}}
	\caption{CPU time for PAC-ized sampling methods via {\aaa} setting $\epsilon=\delta=0.2$ ({\allterminal}).}
	\label{fg:grids_paced}
\end{figure}

On the other hand, ${\gbasf}(Y^{CMC})$ turned out to be practical for more cases, and the analysis used by Huber~\cite{Huber2017} seems to be tight as evidenced by our estimates of {\deltao} (Figure~\ref{fg:cmc}, a-b). Yet, as expected, the running time is heavily penalized by a factor $1/\unrelf$ in the expected sample size as shown in Figure~\ref{fg:cmc}~(c).
\begin{figure}
	\centering
	\subfloat[{\twoterminal}]{\includegraphics[width=0.33\textwidth]{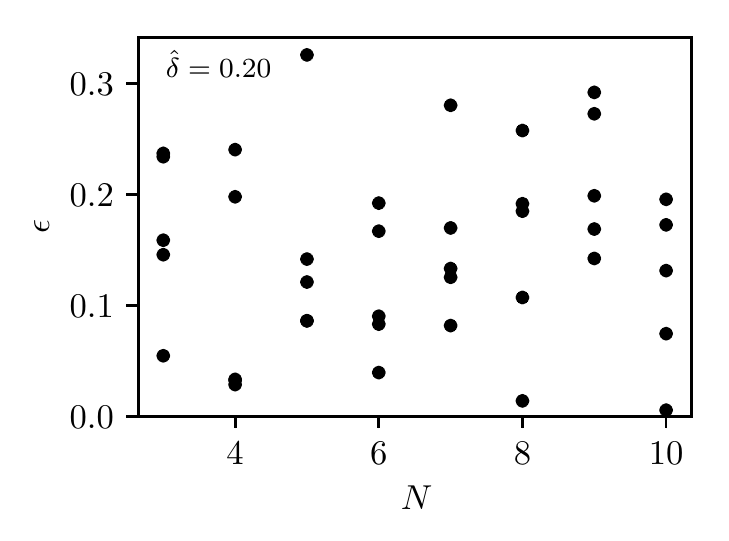}}\hfill
	\subfloat[{{\kterminal}}]{\includegraphics[width=0.33\textwidth]{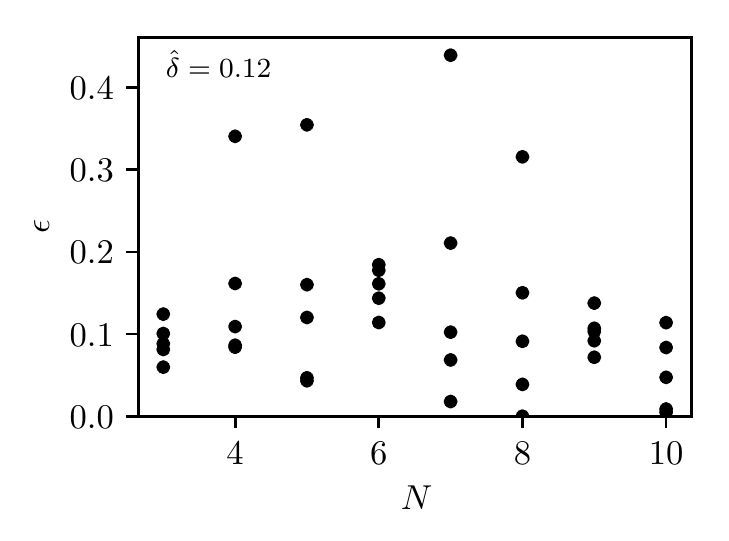}}\hfill
	\subfloat[{{\allterminal}}]{\includegraphics[width=0.33\textwidth]{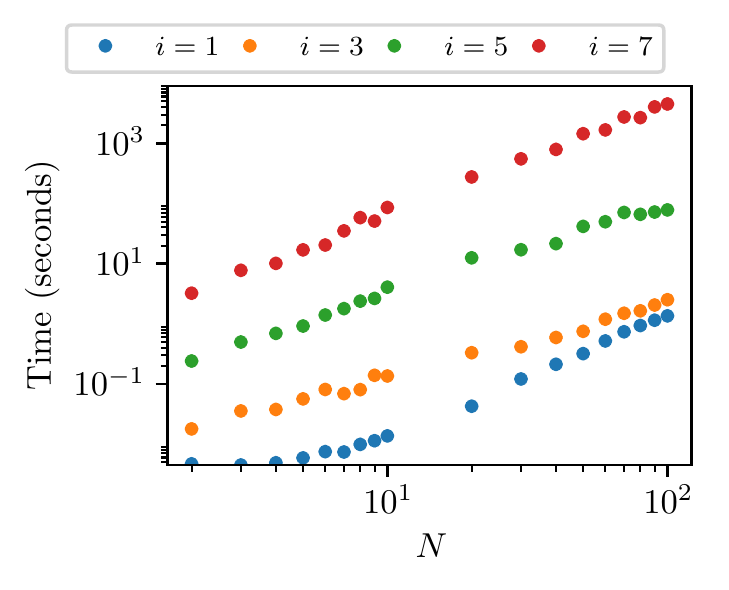}}\hfill
	\caption{ (a-b) Multiplicative error for $\gbasf(Y^{CMC})$ setting $\epsilon=\delta=0.2$, and (c) respective running time for various sizes.}
	\label{fg:cmc}
\end{figure}

Furthermore, we used {\relnet} to approximate {\unrel} in all cases of $\mathcal{K}$ thanks to our new developments. Figure~\ref{fg:grids_pac}(b) shows runtimes as well as $(\deltaof,\epsof)$ values for edge failure probability cases of $2^{-1},2^{-3},2^{-5}$. The weighted to unweighted transformation appears to be the current bottleneck as it considerably increases the number of extra variables in $F_{\mathcal{K}}$. However, note that, unlike K2Simple that is specialized for the all-terminal case, {\relnet} is readily applicable to any {\kterminal} problem instance. Also, {\relnet} is the only method that, due to its dependence on an external Oracle, can exploit on-going third-party developments, as constrained SAT and weighted model counting are very active areas of research.~\footnote{See past and ongoing competitions: \href{https://www.satcompetition.org/}{https://www.satcompetition.org/}} Also, SAT-based methods are uniquely positioned to exploit breakthroughs in quantum hardware and support a possible quantum version of {\relnet}~\cite{Duenas-Osorio2017b}.
\begin{figure}
	\centering
	\subfloat[K2Simple ($\epsilon=0.2$ and $\delta=0.05$)]{
		\begin{minipage}[t][][t]{.50\textwidth}
			\centering
			\includegraphics[width=0.49\textwidth]{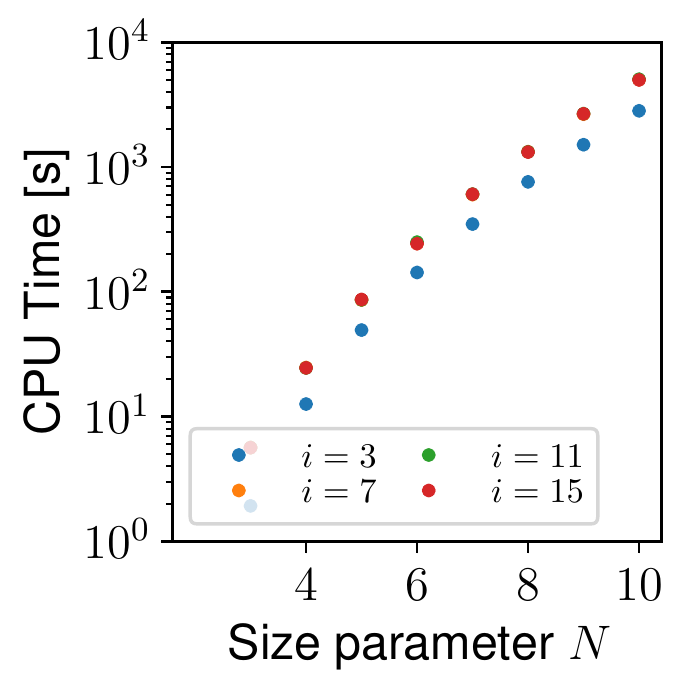}
			\includegraphics[width=0.49\textwidth]{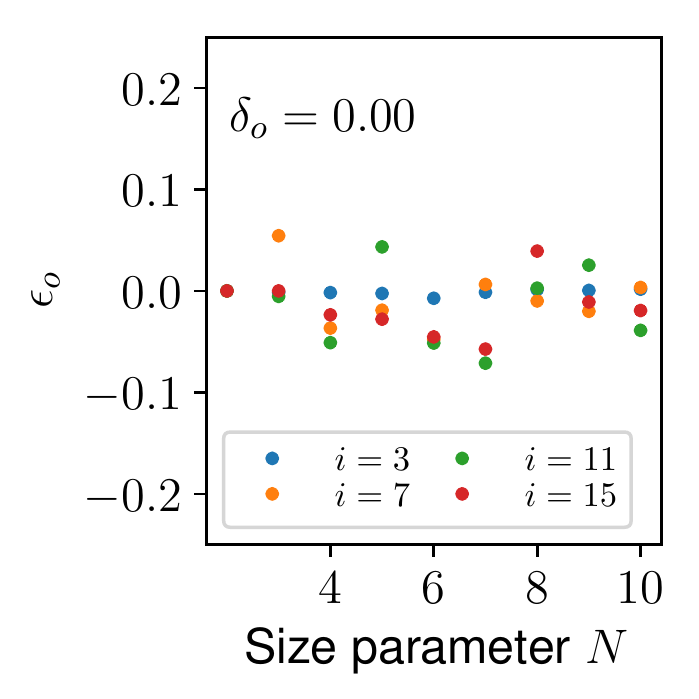}
		\end{minipage}%
	}
	\subfloat[{\relnet} ($\epsilon=0.8$ and $\delta=0.2$)]{
		\begin{minipage}[t][][t]{.50\textwidth}
			\centering
			\includegraphics[width=0.49\textwidth]{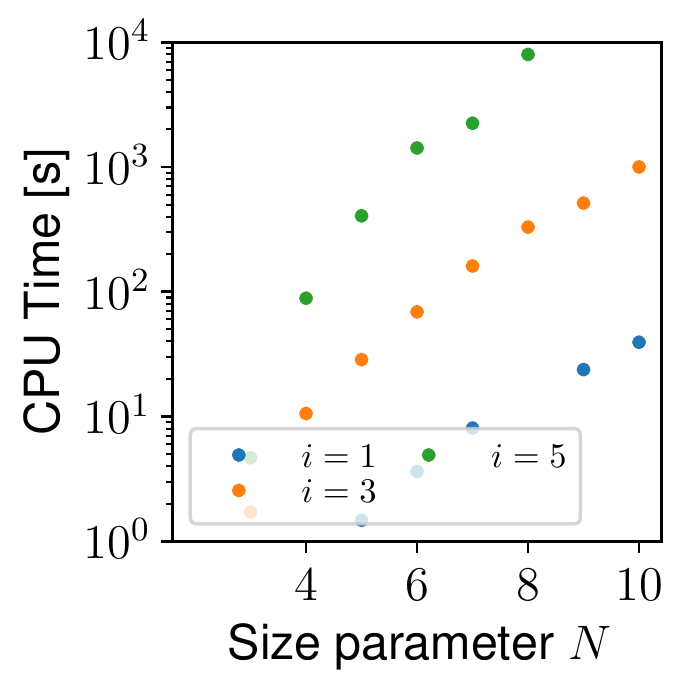}
			\includegraphics[width=0.49\textwidth]{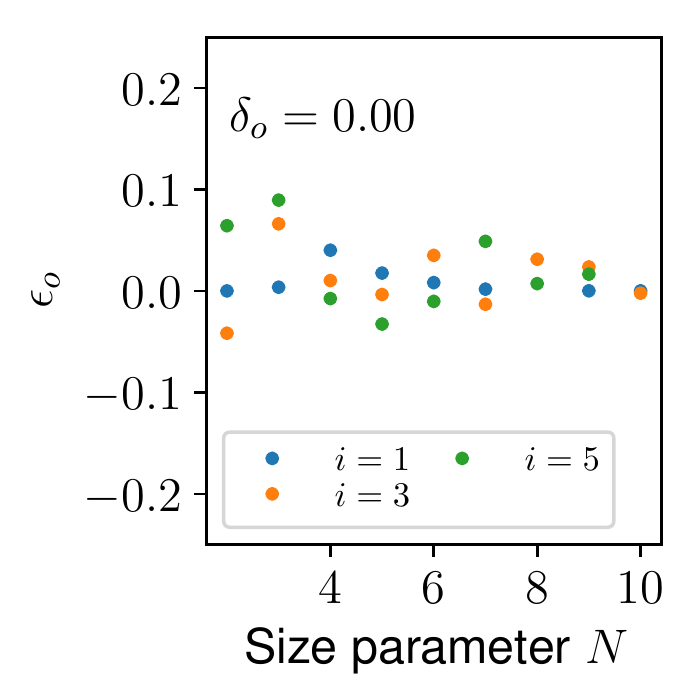}
		\end{minipage}
	}
	\caption{{\relnet} CPU time and {\epso} values for {\kterminal} case.}
	\label{fg:grids_pac}
\end{figure}

Furthermore, our experimental results suggest that the analysis of both, K2Simple and {\relnet}, is not tight. This is observed by values of $(\epsof,\deltaof)$, which are far better than the theoretical input guarantees. This calls for further refinement in their theoretical analysis. Conversely, GBAS delivers practical guarantees that are much closer to the theoretical ones, as demonstrated in Figures~\ref{fg:cmc} and \ref{fg:grids_gbas}, where the target error can be exceeded still satisfying the target confidence overall.
\begin{figure}
	\centering
	\subfloat[All-terminal]{\includegraphics[width=0.25\textwidth]{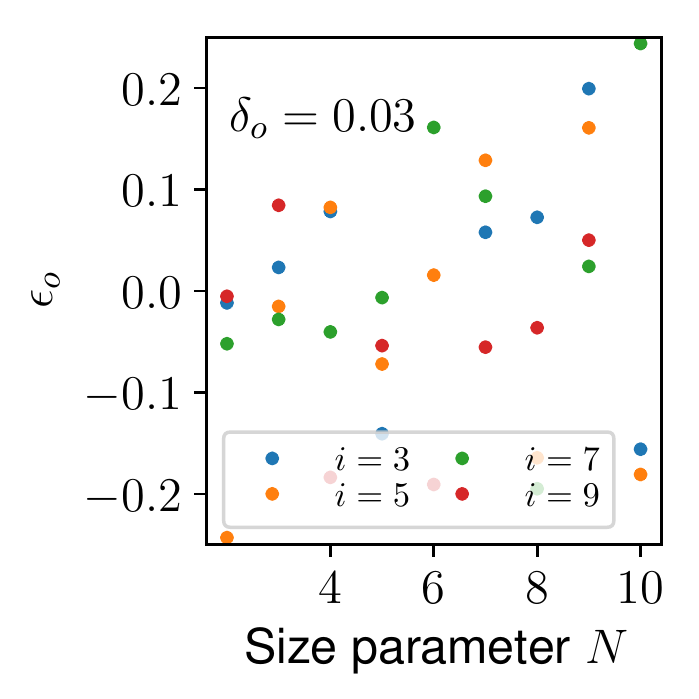}}
	\subfloat[Two-terminal]{\includegraphics[width=0.25\textwidth]{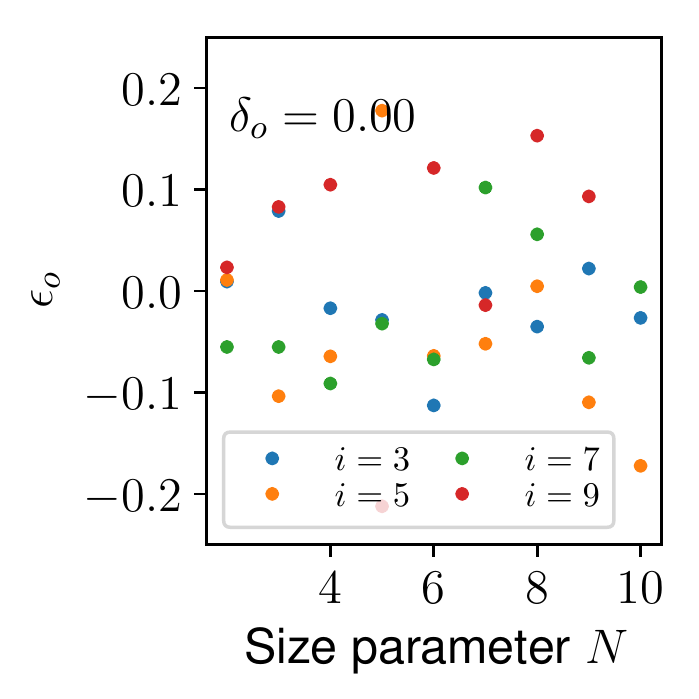}}
	\subfloat[$\mathcal{K}$-terminal]{\includegraphics[width=0.25\textwidth]{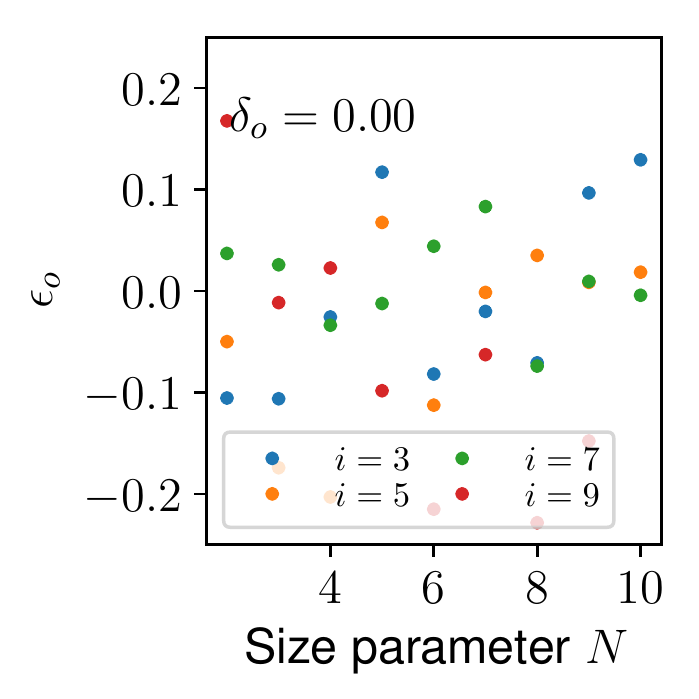}}
	\subfloat[Two-terminal]{\includegraphics[width=0.25\textwidth]{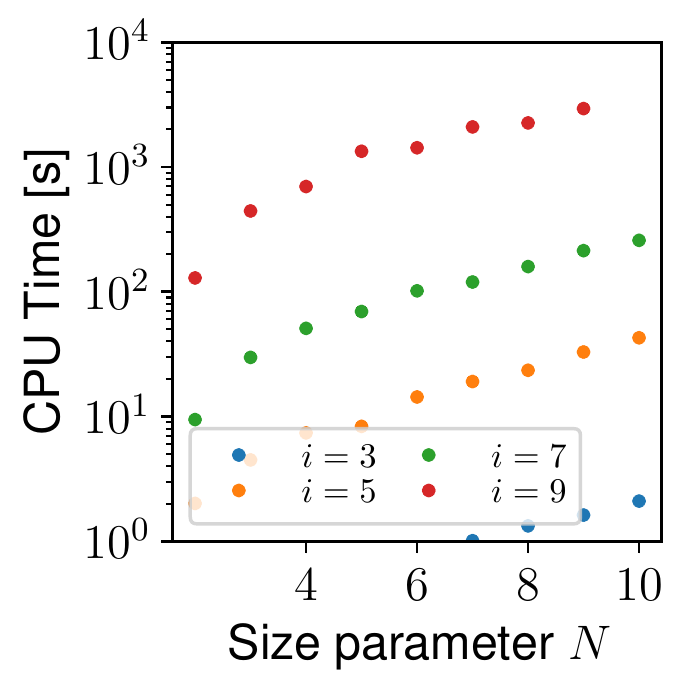}}
	\caption{{\epso} values and CPU time for ${\gbasf(Y^{CMC})}$ setting $(\epsilon,\delta)=(0.2,0.05)$.}
	\label{fg:grids_gbas}
\end{figure}

The square grids gave us insight on the relative performance of reliability estimation methods. Next, we use a dataset of power transmission networks to test methods on instances with engineered system topologies.

\subsection{U.S. Power Transmission Networks}
We consider a dataset with 58 power transmission networks in cities across the U.S. A summary discussion of their structural graph properties can be found elsewhere~\cite{Li201684}. Also, we considered the two-terminal reliability problem. To test the robustness of methods, for each instance {\instance}, we considered every possible $s,t\in V$ pair as a different experiment. Thus, totaling $\binom{n}{2}$ experiments per network instance, where $n=|V|$. We used a single edge failure probability across experiments of $p_e=2^{-3}=0.125$ to keep overall computation time practical. Using {\hll} and preprocessing of networks, we were able to get exact estimates for some of the experiments. We used these to measure the observed multiplicative error {\epso} when possible. Computational times are reported for all experiments, even if multiplicative error is unknown.
\begin{figure}
	\centering
	\subfloat[Multiplicative error {\epso}]{\includegraphics[width=0.5\textwidth]{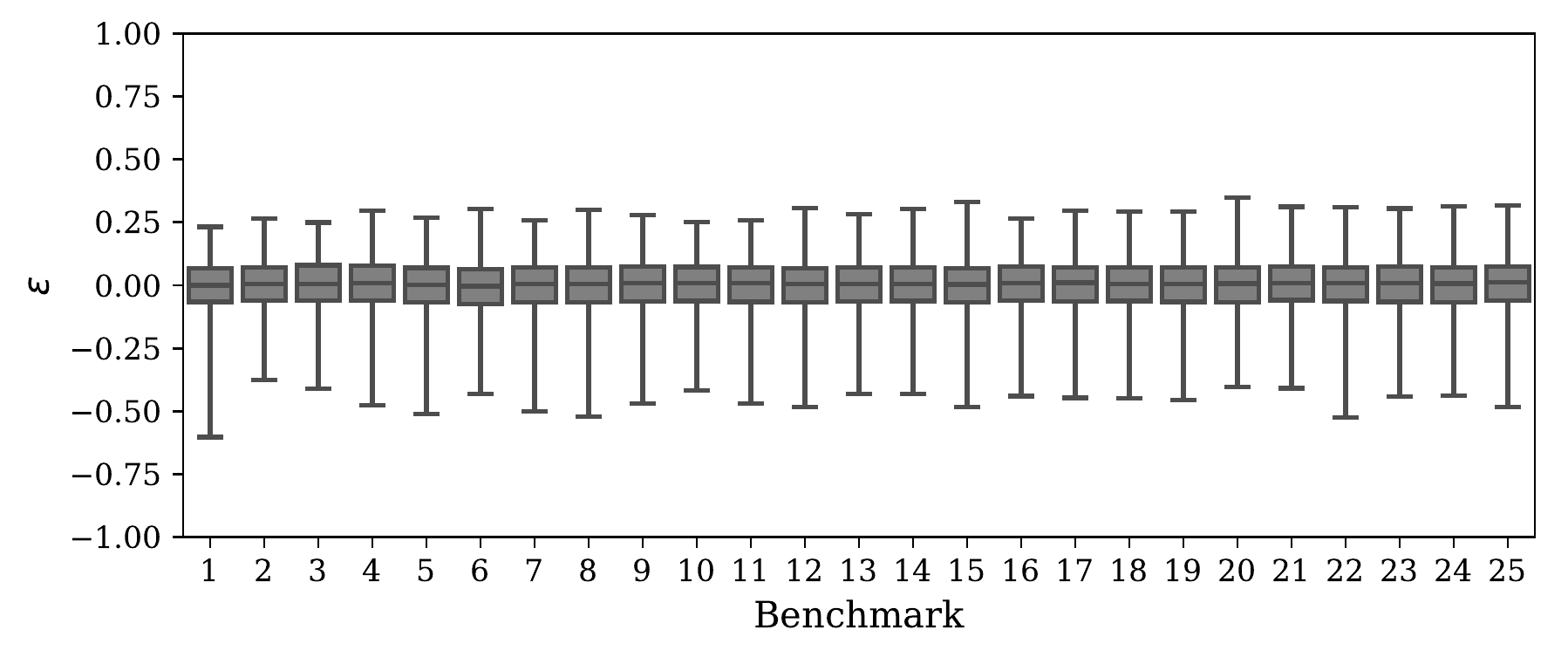}}\hfill
	\subfloat[Running time (seconds)]{\includegraphics[width=0.5\textwidth]{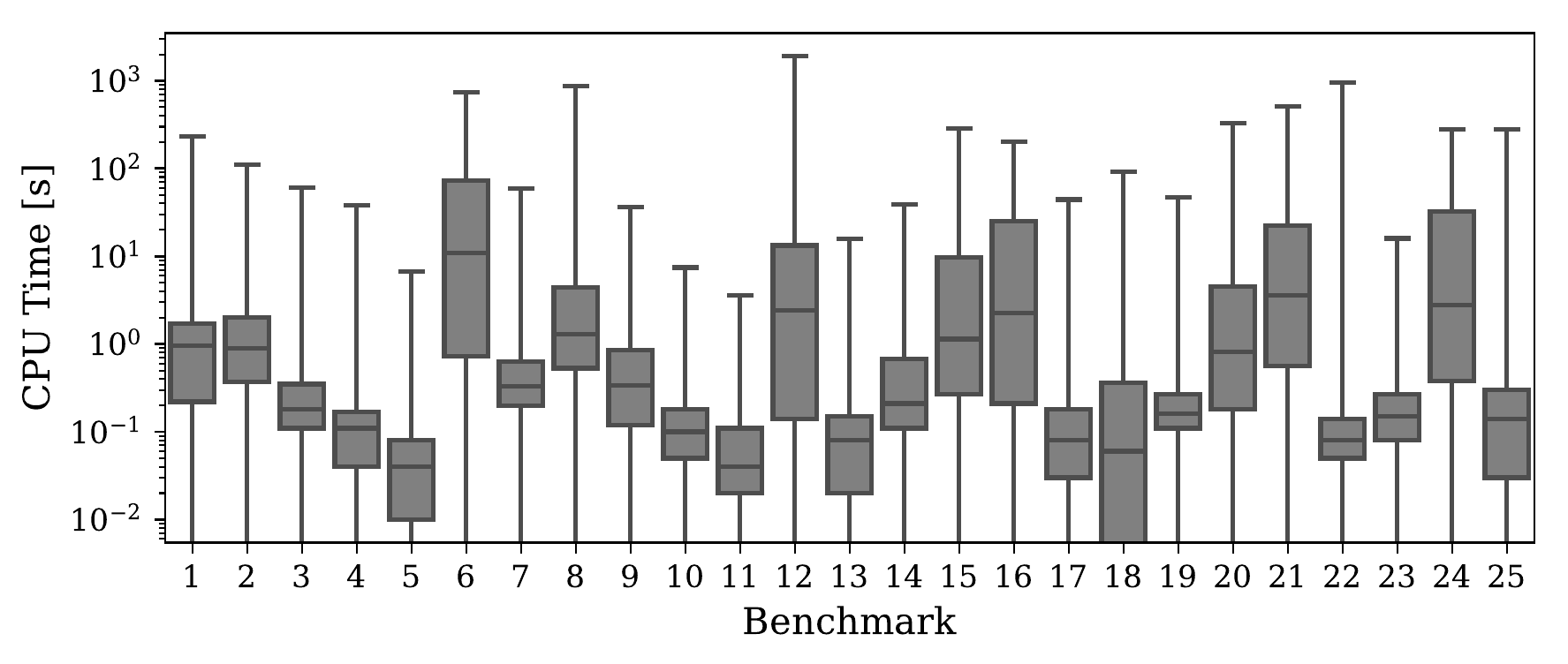}}\hfill
	\caption{ Two-terminal reliability approximations using {\gbas} setting $\epsilon=\delta=0.2$}
	\label{fg:pt_cmc}
\end{figure}

Figure~\ref{fg:pt_cmc} shows PAC-estimates using {\gbas}. As expected, the variation in CPU time was proportional to $1/{\unrelf}$. Furthermore, we  used {\relnet} to obtain PAC-estimates and observed consistent values of the multiplicative error (Figure~\ref{fg:pt_relnet}). In some instances, however, {\relnet} failed to return an estimate before timeout.
We also tested simulation methods setting $N_S=10^3$. Despite the lack of guarantees they performed well in terms of {\epso} and CPU time~(Figures~\ref{fg:pt_sim_epsilon}-\ref{fg:pt_sim_time}, first 5 benchmarks for brevity). However, the efficiency ratio is reduced as the size of instances grows.
\begin{figure}
	\centering
	\subfloat[Multiplicative error {\epso}]{\includegraphics[width=0.5\textwidth]{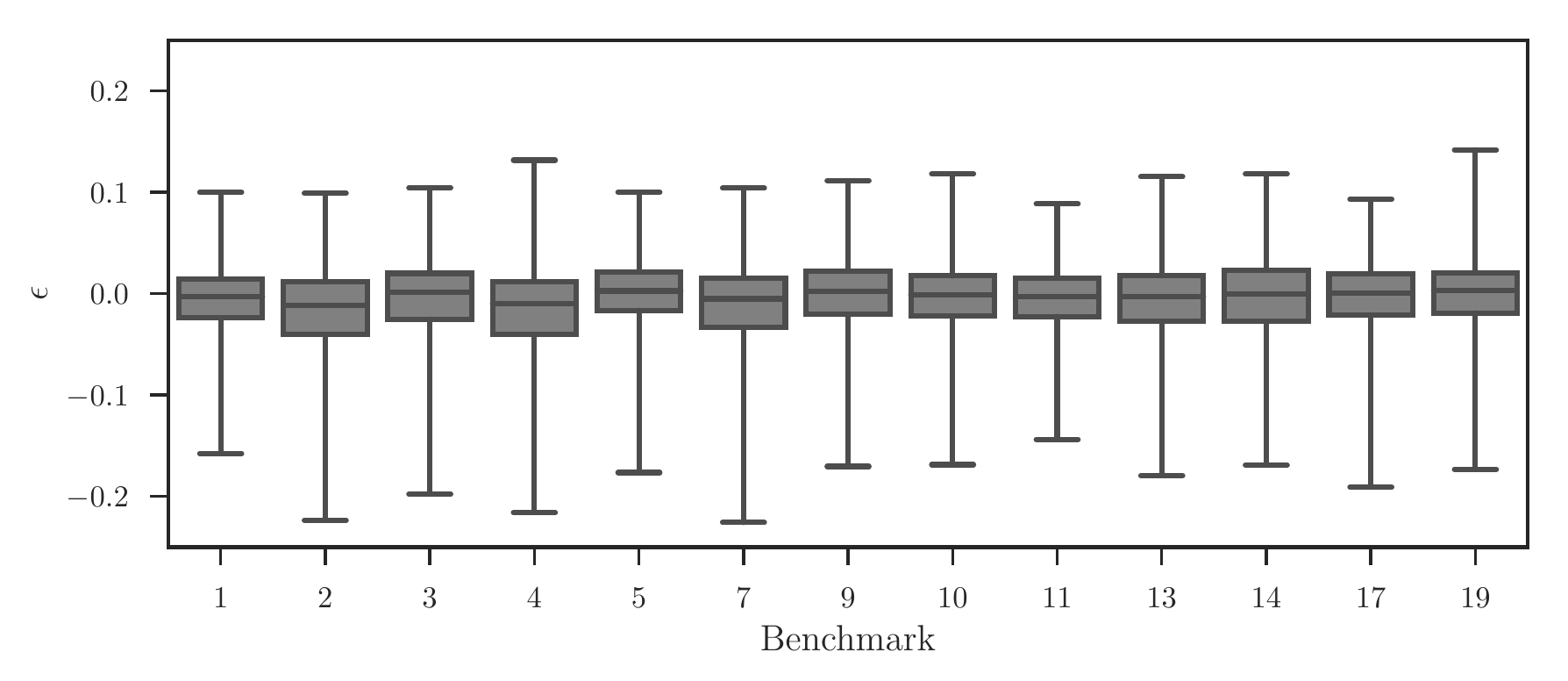}}\hfill
	\subfloat[Running time (seconds)]{\includegraphics[width=0.5\textwidth]{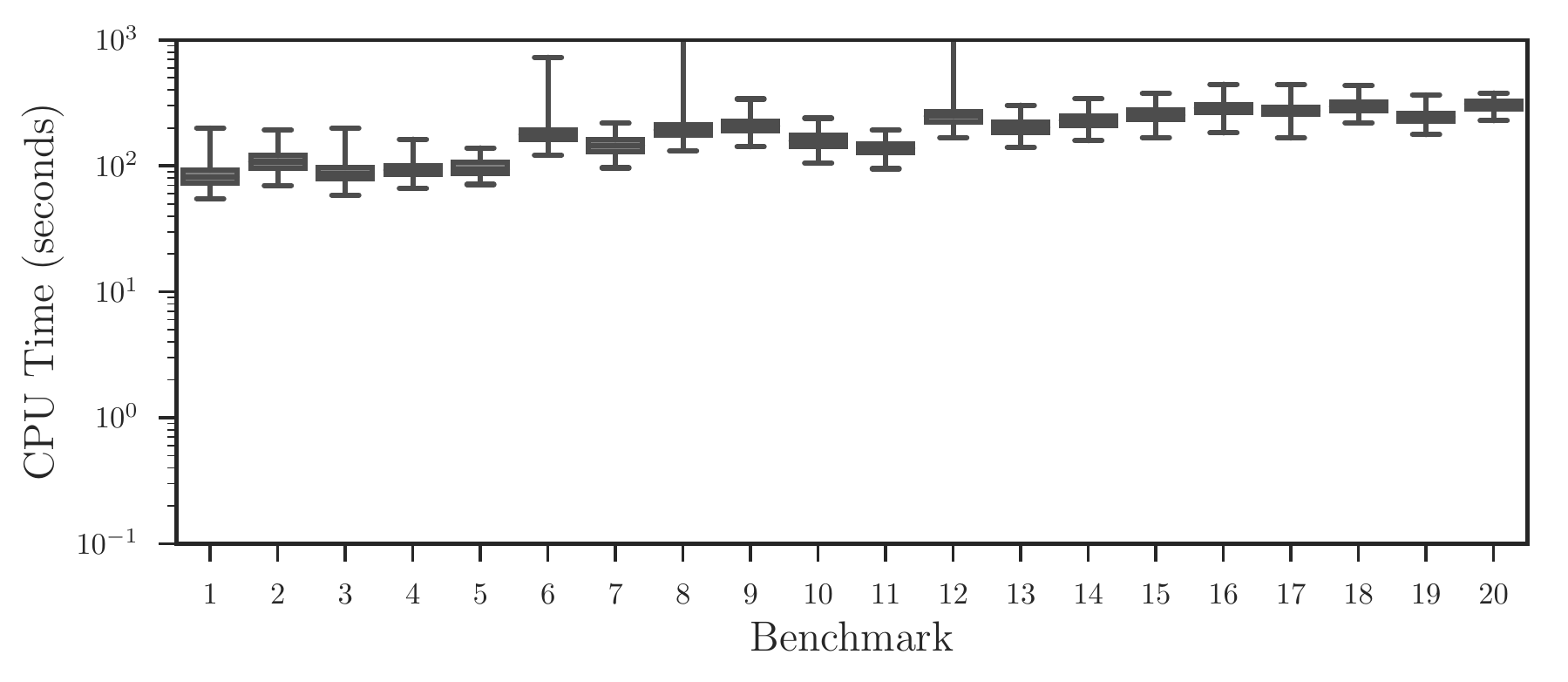}}\hfill
	\caption{Two-terminal reliability approximations using {\relnet} with $(0.8, 0.2)$.}
	\label{fg:pt_relnet}
\end{figure}

\begin{figure}
	\centering
	\subfloat[LT]{\includegraphics[width=0.25\textwidth]{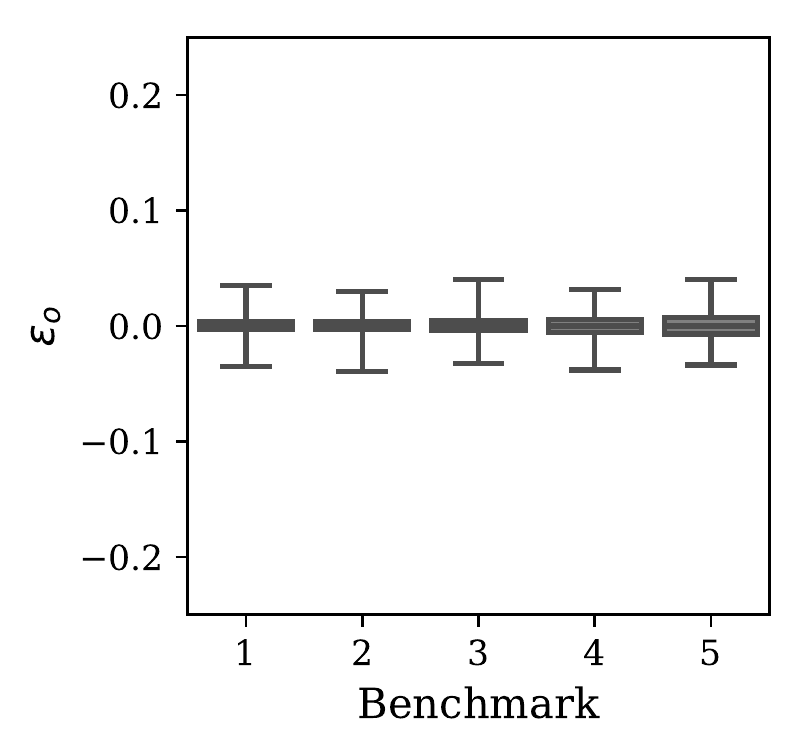}}\hfill
	\subfloat[ST]{\includegraphics[width=0.25\textwidth]{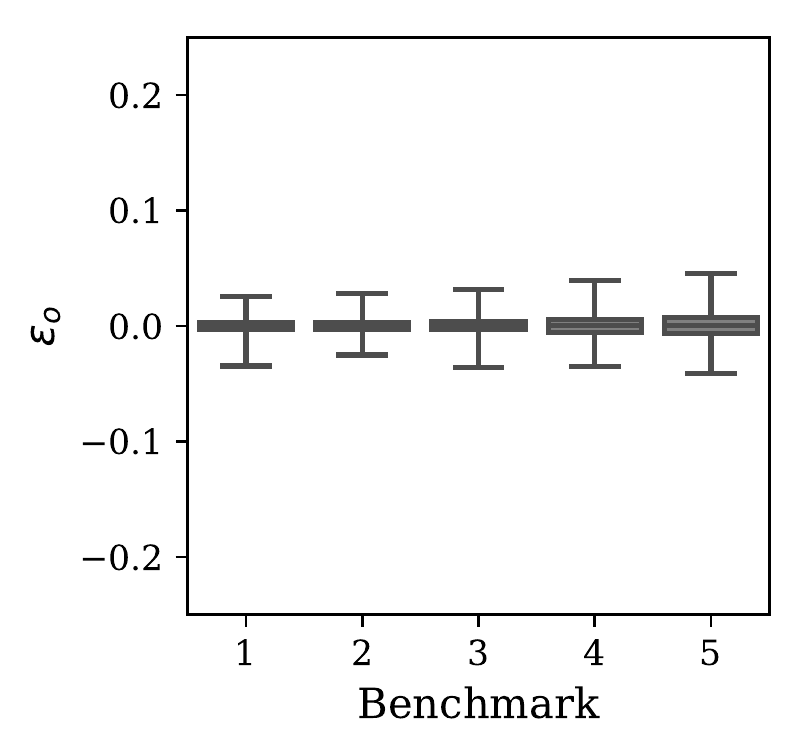}}\hfill
	\subfloat[GS]{\includegraphics[width=0.25\textwidth]{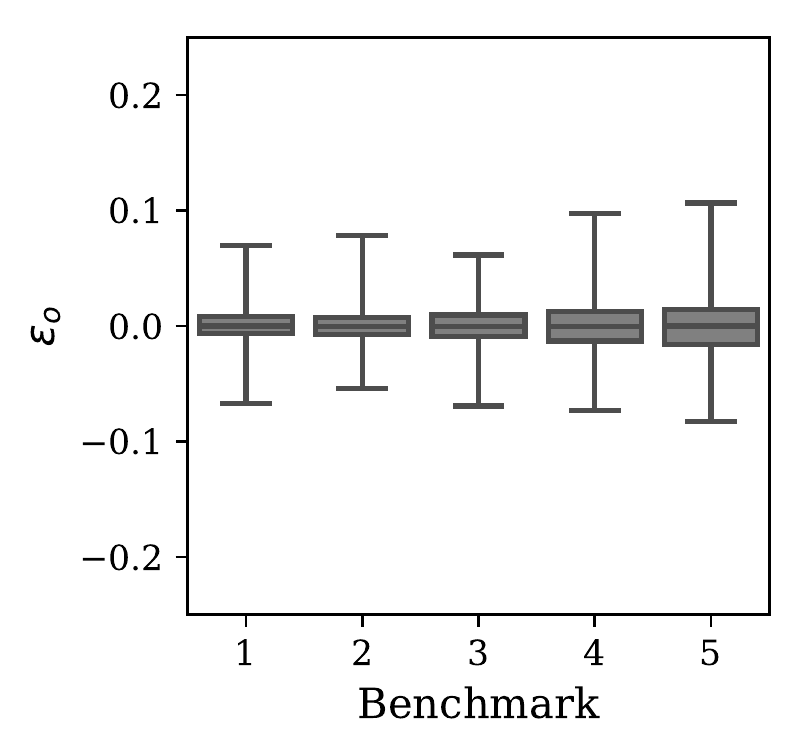}}\hfill
	\subfloat[RVR]{\includegraphics[width=0.25\textwidth]{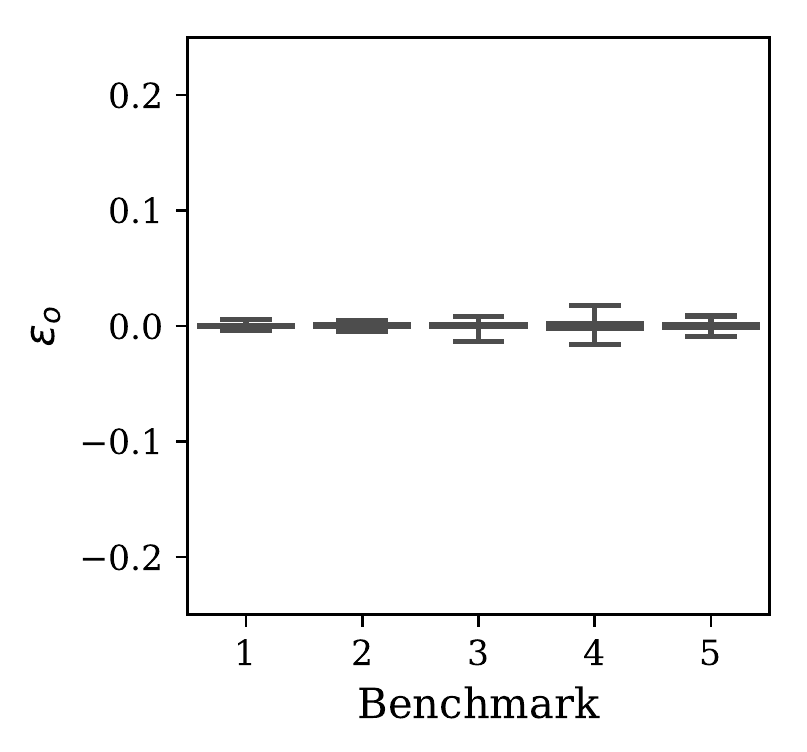}}\hfill
	\caption{Multiplicative error for simulation methods setting $N_S=10^3$.}
	\label{fg:pt_sim_epsilon}
\end{figure}

\begin{figure}
	\centering
	\subfloat[LT]{\includegraphics[width=0.25\textwidth]{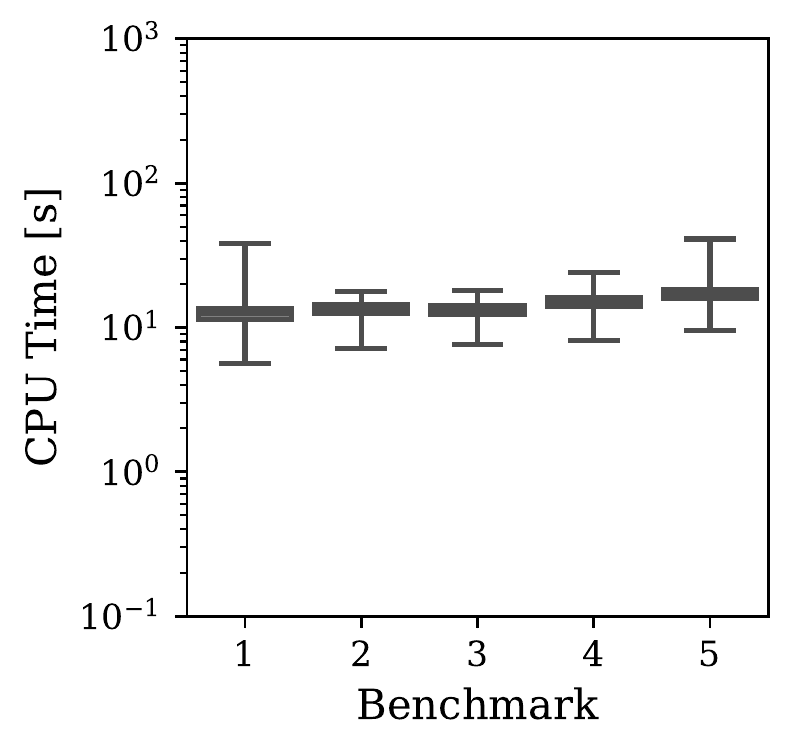}}\hfill
	\subfloat[ST]{\includegraphics[width=0.25\textwidth]{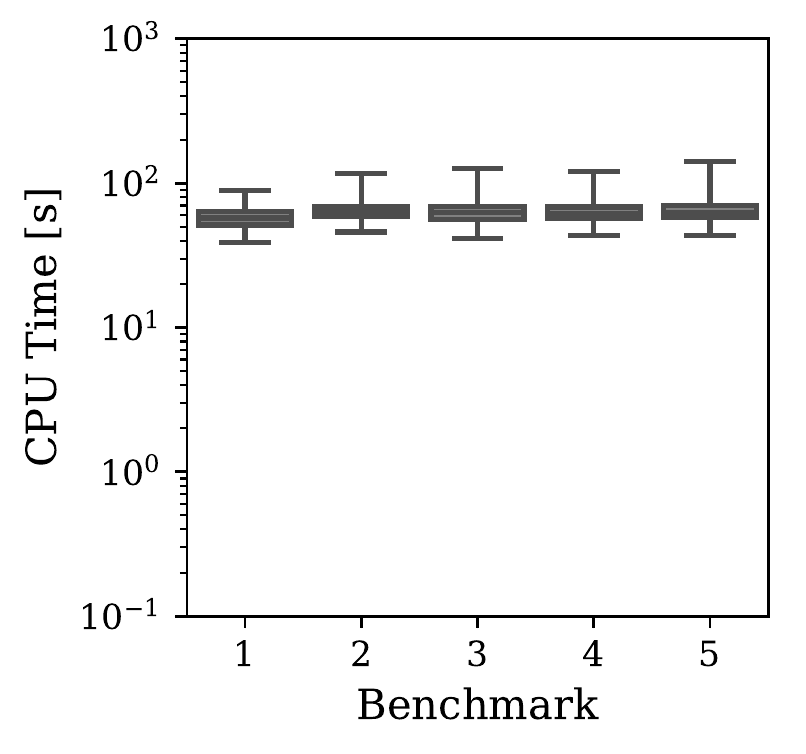}}\hfill
	\subfloat[GS]{\includegraphics[width=0.25\textwidth]{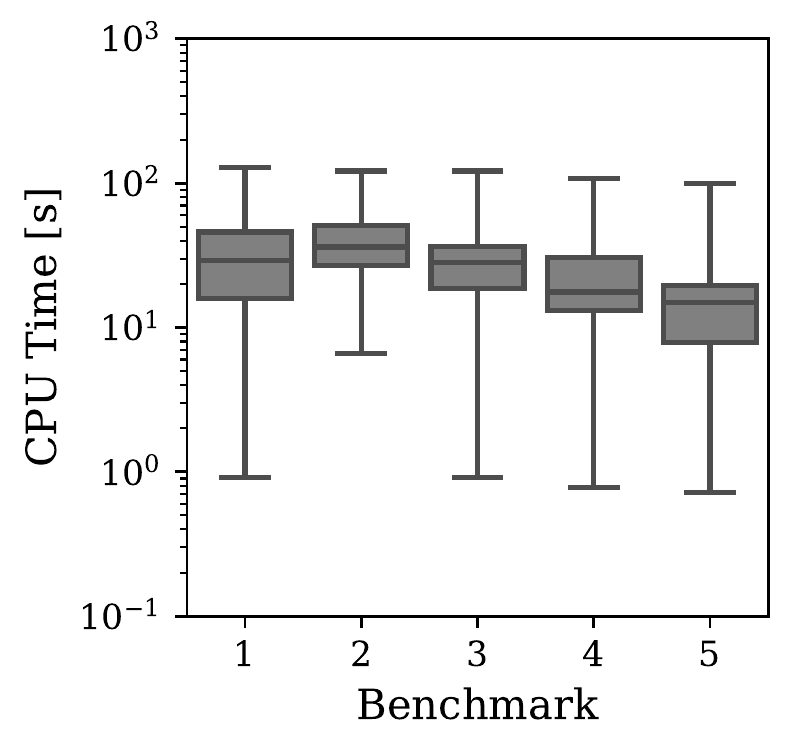}}\hfill
	\subfloat[RVR]{\includegraphics[width=0.25\textwidth]{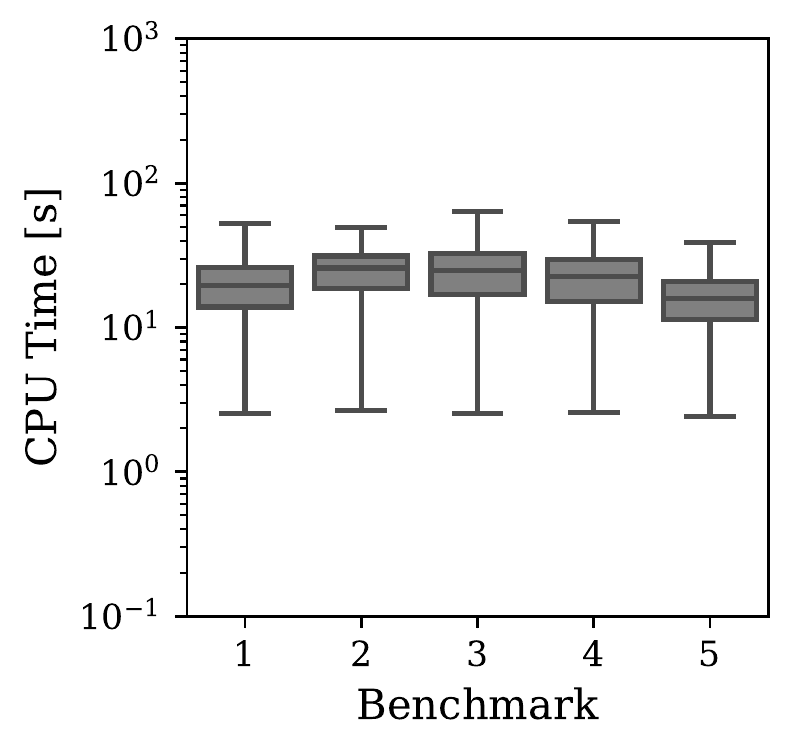}}\hfill
	\caption{Running time for simulation methods setting $N_S=10^3$.}
	\label{fg:pt_sim_time}
\end{figure}

\subsection{Analysis of results and outlook}

Exact methods are advantageous when a topological property is known to be bounded. {\hll} proved useful not only for medium-sized grids ($N\times N=100$), but also it was instrumental when computing exact estimates for many streamlined power transmission networks. Our research shows that methods exploiting bounded properties, together with practical upper bounds, deliver competitive exact calculations for many engineered systems. In power transmission networks, {\hll} was able to exploit their relatively small treewidth.

Among guarantee-less sampling methods, there are multiple paths for improvement. In the cases of LT and ST methods, even when the exponential matrix offers a reliable approach to compute the convolution of exponential random variables, numerically stable computations represent the main bottleneck of the algorithms and in many cases, they are not needed. Thus, future research could devise ways to diagnose these issues and fall-back to the exponential matrix only when needed, or use approximate integration (as in Gertsbakh et al.~\cite{Gertsbakh2015}), or use a more arithmetically robust algorithm (e.g. round-off algorithms for the sample variance~\cite{Chan1983}). Moreover, GS was competitive but its requirement to run a preliminary experiment with an arbitrary number of trajectories $N_0$ to define intermediate levels, and without a formal guidance on its values, can represent a practical barrier when there is no knowledge in the order of magnitude of {\unrel}. Future research could devise splitting mechanisms that use all samples towards the final experiments while retaining its unbiased properties. Finally, RVR was very competitive; however, we noted that (i) the number of terminals adds a considerable overhead in the number of calls to the minimum cut algorithm, and (ii) its performance is tied to the number of maximum probability cuts because larger cuts do not contribute meaningfully towards computing {\unrel}. Future work could use Karger's minimum cut approximation~\cite{Karger1996} and an adaptive truncation of the recursion found in the RVR estimator to address (i) and (ii), respectively. We are currently investigating this very issue and recognized the RVR estimator as an special, yet randomized, case of state-space-partition algorithms~\cite{Paredes2018}.

Among PAC-methods, we found {\gbas} to be tight in its theoretical analysis and competitive in practice. Outside the extremely rare-event regime, we contend that the usage of PAC algorithms such as GBAS would benefit the reliability and system safety community as they give exact confidence intervals without the need of asymptotic assumptions and arbitrary choices on the number of samples and replications. Karger's newly suggested algorithms demonstrated practical performance even in the rare-event regime, yet it appears that their theoretic guarantees are still too conservative. Equipping K2Simple with {\gbas} at the first recursion level would instantly yield a faster algorithm for non-small failure probabilities. However, the challenge of proving tighter bounds on the relative variance for the case of small failure probabilities remains. The same argument on theoretic guarantees being too conservative extends to RelNet, which cannot be set too tight in practice. But we expect RelNet to gain additional competitiveness as orthogonal advances in approximate weighted model counting continue to accrue. RelNet remains competitive in the non rare-event regime, delivering rigorous PAC-guarantees for the {\kterminal} problem. Also, its SAT-based formulation makes it uniquely suitable for quantum algorithmic developments, at a time when major technological developers, such as IBM, Google, Intel, etc., are increasing their investment on quantum hardware~\cite{ preskill2018quantum}.

\section{Conclusions and Future Work}
\noindent We introduced a new logic-based method for the {\kterminal} problem, {\relnet}, which offers rigorous guarantees on the quality of its approximations. We examined this method relative to several other competitive approaches. For non-exact methods we emphasized desired relative variance properties: bounded by a polynomial on the size of the input instance (FPRAS), bounded by a constant (BRV), or tending to zero (VRV). We turned popular estimators in the literature into probably approximately correct (PAC) ones by embedding them into an optimal Monte Carlo algorithm, and showed their practical performance using a set of benchmarks. 

Our tool, {\relnet}, is the first approximation of the {\kterminal} problem, giving strong performance guarantees in the FPRAS sense (relative to a SAT-oracle). Also, {\relnet} gives rigorous multiplicative error guarantees, which are more conservative than relative error guarantees. However, its performance in practice remains constrained to not too small edge failure probabilities ($\approx0.1$), which remains practical when conditioned on catastrophic hazard events. Thus, our future work will pursue a more efficient encoding and solution approaches, especially when edge failure probabilities become smaller. Moreover, promising advances in approximate model counting and SAT solvers will render {\relnet} more efficient over time, given its reliance on SAT oracles.

Embedding estimators with desired relative variance properties into PAC methods proved to be an effective strategy in practice, but only when failure probabilities are not rare.
Despite this relative success, the strategy becomes impractical when {\unrel} approaches zero. Thus, future research can address these issues in two fronts: (i) establishing parameterized upper bounds on the relative variance of new and previous estimators when they exist, and (ii) develop new PAC-methods with faster convergence guarantees than those of the canonical Monte Carlo approach.

PAC-estimation is a promising yet developing approach to system reliability estimation. Beyond the {\kterminal} problem, its application can be challenging in the rare event regime, but in all other cases it can be used much more frequently as an alternative to the less rigorous---albeit pervasive---empirical study of the variance through replications and asymptotic assumptions appealing to the central limit theorem. In fact, methods such as {\gbas} deliver exact confidence intervals using all samples at the user's disposal. In future work, the authors will explore general purpose PAC-methods that can be employed in the rare-event regime, developing a unified framework to conduct reliability assessments with improved knowledge of uncertainties and further promote engineering resilience and align it with the measurement sciences.



\section*{Acknowledgments}
\noindent The authors gratefully acknowledge the support by the U.S. Department of Defense (Grant W911NF-13-1-
0340) and the U.S. National Science Foundation (Grants CMMI-1436845 and CMMI-1541033).

\biboptions{numbers,sort&compress}
\bibliographystyle{elsarticle-num}
\bibliography{biblist}

\begin{thebibliography}{10}
\expandafter\ifx\csname url\endcsname\relax
  \def\url#1{\texttt{#1}}\fi
\expandafter\ifx\csname urlprefix\endcsname\relax\def\urlprefix{URL }\fi
\expandafter\ifx\csname href\endcsname\relax
  \def\href#1#2{#2} \def\path#1{#1}\fi

\bibitem{Ball1986}
M.~O. Ball, {Computational Complexity of Network Reliability Analysis: An
  Overview}, IEEE Transactions on Reliability 35~(3) (1986) 230--239.

\bibitem{Valiant1979}
L.~G. Valiant, {The Complexity of Enumeration and Reliability Problems}, SIAM
  Journal on Computing 8~(3) (1979) 410--421.

\bibitem{Hardy2007}
G.~Hardy, C.~Lucet, N.~Limnios, {K-Terminal Network Reliability Measures With
  Binary Decision Diagrams}, IEEE Transactions on Reliability 56~(3) (2007)
  506--515.

\bibitem{Canale2016}
E.~Canale, P.~Romero, G.~Rubino, {Factorization theory in diameter constrained
  reliability}, in: 2016 8th International Workshop on Resilient Networks
  Design and Modeling (RNDM), Vol.~6, IEEE, 2016, pp. 66--71.

\bibitem{Fishman1986}
G.~S. Fishman, {A Monte Carlo Sampling Plan for Estimating Network
  Reliability}, Operations Research 34~(4) (1986) 581--594.

\bibitem{Valiant1984}
L.~G. Valiant, {A theory of the learnable}, Communications of the ACM 27~(11)
  (1984) 1134--1142.

\bibitem{Fishman1996}
G.~S. Fishman, {Monte Carlo: Concepts, Algorithms, and Applications}, Springer
  New York, New York, NY, 1996.

\bibitem{Gertsbakh2016}
I.~B. Gertsbakh, Y.~Shpungin, {Models of Network Reliability: Analysis,
  Combinatorics, and Monte Carlo}, CRC Press, 2016.

\bibitem{Botev2012}
Z.~I. Botev, P.~L'Ecuyer, G.~Rubino, R.~Simard, B.~Tuffin, {Static Network
  Reliability Estimation via Generalized Splitting}, INFORMS Journal on
  Computing 25~(1) (2012) 56--71.

\bibitem{Zuev2015}
K.~M. Zuev, S.~Wu, J.~L. Beck, {General network reliability problem and its
  efficient solution by Subset Simulation}, Probabilistic Engineering Mechanics
  40 (2015) 25--35.

\bibitem{Cancela1995}
H.~Cancela, M.~{El Khadiri}, {A recursive variance-reduction algorithm for
  estimating communication-network reliability}, IEEE Transactions on
  Reliability 44~(4) (1995) 595--602.

\bibitem{Gertsbakh2010}
I.~B. Gertsbakh, Y.~Shpungin, {Models of Network Reliability: Analysis,
  Combinatorics, and Monte Carlo}, CRC Press, 2010.

\bibitem{Vaisman2016}
R.~Vaisman, D.~P. Kroese, I.~B. Gertsbakh, {Splitting sequential Monte Carlo
  for efficient unreliability estimation of highly reliable networks},
  Structural Safety 63 (2016) 1--10.

\bibitem{Karger2001}
D.~R. Karger, {A Randomized Fully Polynomial Time Approximation Scheme for the
  All-Terminal Network Reliability Problem}, SIAM Review 43~(3) (2001)
  499--522.

\bibitem{Stockmeyer1983}
L.~Stockmeyer, {The complexity of approximate counting}, in: Proceedings of the
  fifteenth annual ACM symposium on Theory of computing - STOC '83, no.~1, ACM
  Press, New York, New York, USA, 1983, pp. 118--126.

\bibitem{Bayer2014}
C.~Bayer, H.~Hoel, E.~Von~Schwerin, R.~Tempone, On nonasymptotic optimal
  stopping criteria in monte carlo simulations, SIAM Journal on Scientific
  Computing 36~(2) (2014) A869--A885.

\bibitem{Ellingwood2016}
B.~R. Ellingwood, J.~W. van~de Lindt, T.~P. McAllister, {Developing measurement
  science for community resilience assessment}, Sustainable and Resilient
  Infrastructure 1~(3-4) (2016) 93--94.

\bibitem{Duenas-Osorio2017}
L.~Duenas-Osorio, K.~S. Meel, R.~Paredes, M.~Y. Vardi, {Counting-based
  Reliability Estimation for Power-Transmission Grids}, in: Proceedings of AAAI
  Conference on Artificial Intelligence (AAAI), San Francisco, 2017.

\bibitem{SM19}
M.~Soos, K.~S. Meel, Bird: Engineering an efficient cnf-xor sat solver and its
  applications to approximate model counting, in: Proceedings of AAAI
  Conference on Artificial Intelligence (AAAI), 2019.

\bibitem{Ball1995}
M.~O. Ball, C.~J. Colbourn, J.~S. Provan, {Chapter 11 Network reliability}, in:
  Handbooks in Operations Research and Management Science, Vol.~7, 1995, pp.
  673--762.

\bibitem{Dotson1979}
W.~Dotson, J.~Gobien, {A new analysis technique for probabilistic graphs}, IEEE
  Transactions on Circuits and Systems 26~(10) (1979) 855--865.

\bibitem{Satyanarayana1983}
A.~Satyanarayana, M.~K. Chang, {Network reliability and the factoring theorem},
  Networks 13~(1) (1983) 107--120.

\bibitem{Le2014}
M.~L{\^{e}}, J.~Weidendorfer, M.~Walter, {A novel variable ordering heuristic
  for BDD-based K-Terminal reliability}, Proceedings - 44th Annual IEEE/IFIP
  International Conference on Dependable Systems and Networks, DSN 2014 (2014)
  527--537.

\bibitem{Le2013}
M.~L{\^{e}}, M.~Walter, J.~Weidendorfer, {A memory-efficient bounding algorithm
  for the two-terminal reliability problem}, Electronic Notes in Theoretical
  Computer Science 291 (2013) 15--25.

\bibitem{Lim2012}
H.-W. Lim, J.~Song, {Efficient risk assessment of lifeline networks under
  spatially correlated ground motions using selective recursive decomposition
  algorithm}, Earthquake Engineering {\&} Structural Dynamics 41~(13) (2012)
  1861--1882.

\bibitem{Paredes2018}
R.~Paredes, L.~Due{\~{n}}as-Osorio, I.~Hernandez-Fajardo, {Decomposition
  algorithms for system reliability estimation with applications to
  interdependent lifeline networks}, Earthquake Engineering {\&} Structural
  Dynamics 47~(13) (2018) 2581--2600.

\bibitem{Alexopoulos1995}
C.~Alexopoulos, {A note on state-space decomposition methods for analyzing
  stochastic flow networks}, IEEE Transactions on Reliability 44~(2) (1995)
  354--357.

\bibitem{Cancela2014}
H.~Cancela, M.~{El Khadiri}, G.~Rubino, B.~Tuffin, {Balanced and approximate
  zero-variance recursive estimators for the network reliability problem}, ACM
  Transactions on Modeling and Computer Simulation 25~(1) (2014) 1--19.

\bibitem{Jerrum1988}
M.~Jerrum, A.~Sinclair, {Conductance and the rapid mixing property for Markov
  chains: the approximation of permanent resolved}, in: Proceedings of the
  twentieth annual ACM symposium on Theory of computing - STOC '88, ACM Press,
  New York, New York, USA, 1988, pp. 235--244.

\bibitem{Fishman1994}
G.~S. Fishman, {Markov Chain Sampling and the Product Estimator}, Operations
  Research 42~(6) (1994) 1137--1145.

\bibitem{Dyer1991}
M.~Dyer, A.~Frieze, R.~Kannan, {A random polynomial-time algorithm for
  approximating the volume of convex bodies}, Journal of the ACM 38~(1) (1991)
  1--17.

\bibitem{Glasserman1999}
P.~Glasserman, P.~Heidelberger, {Multilevel Splitting for Estimating Rare Event
  Probabilities}, Operations Research 47~(4) (1999) 585--600.

\bibitem{kahn1951estimation}
H.~Kahn, T.~E. Harris, Estimation of particle transmission by random sampling,
  National Bureau of Standards applied mathematics series 12 (1951) 27--30.

\bibitem{rosenbluth1955monte}
M.~N. Rosenbluth, A.~W. Rosenbluth, Monte carlo calculation of the average
  extension of molecular chains, The Journal of Chemical Physics 23~(2) (1955)
  356--359.

\bibitem{Au2001}
S.~K. Au, J.~L. Beck, {Estimation of small failure probabilities in high
  dimensions by subset simulation}, Probabilistic Engineering Mechanics 16~(4)
  (2001) 263--277.

\bibitem{Karger2016}
D.~R. Karger, {A Fast and Simple Unbiased Estimator for Network
  (Un)reliability}, 2016 IEEE 57th Annual Symposium on Foundations of Computer
  Science (FOCS) (2016) 635--644.

\bibitem{Dagum2000}
P.~Dagum, R.~Karp, M.~Luby, S.~Ross, {An optimal algorithm for Monte Carlo
  estimation}, Proceedings of IEEE 36th Annual Foundations of Computer Science
  29~(92) (1995) 1--22.

\bibitem{Huber2017}
M.~Huber, {A Bernoulli mean estimate with known relative error distribution},
  Random Structures {\&} Algorithms 50~(2) (2017) 173--182.

\bibitem{wald1973sequential}
A.~Wald, {Sequential Analysis}, John Wiley {\&} Sons, New York, 1947.

\bibitem{Herrmann2010}
J.~U. Herrmann, {Improving Reliability Calculation with Augmented Binary
  Decision Diagrams}, in: 2010 24th IEEE International Conference on Advanced
  Information Networking and Applications, IEEE, 2010, pp. 328--333.

\bibitem{Fishman1986c}
G.~S. Fishman, {A Comparison of Four Monte Carlo Methods for Estimating the
  Probability of s-t Connectedness}, IEEE Transactions on Reliability 35~(2)
  (1986) 145--155.

\bibitem{Duenas-Osorio2017b}
L.~Due{\~{n}}as-Osorio, M.~Vardi, J.~Rojo, {Quantum-inspired Boolean states for
  bounding engineering network reliability assessment}, Structural Safety
  75~(April 2017) (2018) 110--118.

\bibitem{Li201684}
J.~Li, L.~Due{\~{n}}as-Osorio, C.~Chen, B.~Berryhill, A.~Yazdani,
  {Characterizing the topological and controllability features of U.S. power
  transmission networks}, Physica A: Statistical Mechanics and its Applications
  453 (2016) 84--98.

\bibitem{Gertsbakh2015}
I.~Gertsbakh, E.~Neuman, R.~Vaisman, {Monte Carlo for Estimating Exponential
  Convolution}, Communications in Statistics - Simulation and Computation
  44~(10) (2015) 2696--2704.

\bibitem{Chan1983}
T.~F. Chan, G.~H. Golub, R.~J. LeVeque, {Algorithms for Computing the Sample
  Variance: Analysis and Recommendations}, The American Statistician 37~(3)
  (1983) 242.

\bibitem{Karger1996}
D.~R. Karger, C.~Stein, {A new approach to the minimum cut problem}, Journal of
  the ACM 43~(4) (1996) 601--640.

\bibitem{preskill2018quantum}
J.~Preskill, Quantum {C}omputing in the {NISQ} era and beyond, {Quantum} 2
  (2018) 79.

\bibitem{Sinclair2018}
A.~Sinclair,
  \href{https://people.eecs.berkeley.edu/~sinclair/cs271/n11.pdf}{{Lecture 11 :
  Counting satisfying assignments of a DNF formula. University of California,
  Berkeley (CS 271)}} (2018).
\newline\urlprefix\url{https://people.eecs.berkeley.edu/~sinclair/cs271/n11.pdf}

\end{thebibliography}


\newpage
\appendix
\section*{Proof of Theorem~\ref{theom:markov}}
The next two lemmas are useful towards proving Theorem~\ref{theom:markov}. Lemma~\ref{lemma:meanvar} shows how the variance in a Monte Carlo estimator reduces as a function of the number of samples.
\begin{lemma}
	\label{lemma:meanvar}
	For $Y$ a random variable with mean $\mu_Y$ and variance $\sigma^2_Y$, define $Y_n$ as follows:
	\begin{equation*}
		Y_n = \frac{1}{n}\sum_{i=1}^{n} Y_i,
	\end{equation*}
	with $Y_i\sim Y$ and $n\in \mathbb{N}$. Then, we have that $\sigma_{Y_n}^2 = \frac{1}{n} \sigma^2_Y$.
\end{lemma}
\begin{proof}
	The variance of $Y_n$ is:
	\begin{align*}
		\sigma^2_{Y_n} &= \Expect[(Y_n - \mu_{Y_n})^2]\\
		&= \Expect[Y_n^2] - \mu_{Y_n}^2.
	\end{align*}
	Note that $\mu_{Y_n} = \mu_Y$, and using the property of linearity in the expectation operator write:
	\begin{align*}
	\sigma^2_{Y_n} &= \frac{1}{n^2}\sum_{i,j}^{n} \Expect\big[Y_i Y_j\big]- \mu_{Y}^2\\
	&= \frac{1}{n^2}\bigg(\sum_{i}^{n} \Expect\big[Y_i^2\big] + \sum_{i,j:i\neq j}^{n} \Expect\big[Y_i Y_j\big]\bigg) - \mu_{Y}^2.
	\end{align*}
	Now, recall that every $Y_i$ is i.i.d as $Y$ and $\sigma^2_Y=\Expect[Y^2]-\mu_Y^2$. Then,
	\begin{align*}
	\sigma^2_{Y_n} &= \frac{1}{n^2}\bigg(n\cdot\Expect[Y^2] + n(n-1)\cdot\mu_Y^2\bigg) - \mu_{Y}^2\\  
	&=\frac{1}{n} \bigg( \Expect[Y^2] + (n-1)\cdot\mu_{Y}^2 - n\cdot \mu_{Y}^2 \bigg)\\
	&=\frac{1}{n} \big(\Expect[Y^2] - \mu_{Y}^2)\\
	&= \frac{1}{n} \sigma^2_Y.
	\end{align*}
\end{proof}
Lemma~\ref{lemma:boost_confidence} shows the link between the number of repetitions of a experiment and the success probability of the majority of repetitions. We will use this argument for constructing a median-based estimate.
\begin{lemma}
\label{lemma:boost_confidence}
Let $X$ be a Bernoulli random variable with success probability $s\in[0,1]$. Define the random variable:
\begin{equation*}
	X_r = \sum_{i=1}^{r} X_i,
\end{equation*}
with $r\in \mathbb{N}$. Then, the probability of at most $\lfloor r/2 \rfloor$ successes is:
\begin{equation*}
	\Pr(X_r \leq r/2) = \sum_{i=0}^{\lfloor r/2 \rfloor} \binom{r}{i} (s)^{i} (1-s)^{r-i}
\end{equation*}
\end{lemma}
\begin{proof}
The proof is straightforward if one realizes that $X_r$ is a Binomial random variable with parameters $s$ and $r$. The desired probability is the cumulative distribution function evaluated at $r/2$.
\end{proof}

Next, we are ready to prove Theorem~\ref{theom:markov}.
\begin{theorem}
\label{theom:markov}
For a random variable $Y$ with mean $\mu_Y$ and variance $\sigma^2_{Y}$, and user specified parameters $\epsilon,\delta\in(0,1)$, it suffices to draw $O(\sigma^2_{Y}/\mu_Y^2 \epsilon^{-2}\log1/\delta)$ i.i.d samples to compute an estimate $\overline{\mu}_Y$ such that:
\begin{equation*}
	\Pr\bigg(\frac{|\overline{\mu}_Y-\mu_Y|}{\mu_Y}\geq \epsilon\bigg) \leq \delta
\end{equation*}
\end{theorem}
\begin{proof}
	From the well known Markov (or Chebyshev) inequality, we can write:
	\begin{equation*}
		\Pr\big(|Y-\mu_Y|\geq k\big) \leq \frac{\sigma^2_{Y}}{k^2}.
	\end{equation*}
	For our purposes, we let $k=\epsilon\mu_Y$ with positive $\mu_Y$. Then, we write:
	\begin{equation*}
	\Pr\bigg(\frac{|Y-\mu_Y|}{\mu_Y}\geq\epsilon\bigg) \leq \frac{\sigma^2_{Y}}{\epsilon^2\mu_Y^2}.
	\end{equation*}
	If we substitute $Y$ by $Y_n$ such that $n=\frac{\sigma^2_Y}{(1-s)\epsilon^{2}\mu_Y^{2}}$ (Lemma~\ref{lemma:meanvar}), then:
	\begin{equation*}
	\Pr\bigg(\frac{|Y_n-\mu_Y|}{\mu_Y}\geq\epsilon\bigg) \leq \frac{\sigma^2_{Y}/n}{\epsilon^2\mu_Y^2} = 1 - s.
	\end{equation*}
	Since the experiment's success probability is at least $s$, we boost it up to $1-\delta$ via Lemma~\ref{lemma:boost_confidence}. First, let $\overline{\mu}_Y$ be the median of $r$ samples of $Y_n$. Then, note that estimate $\overline{\mu}_Y$ ``fails''---lays outside the interval $\mu_Y(1\pm\epsilon)$---if and only if $r/2$ or more samples lay outside $\mu_Y(1\pm\epsilon)$. Thus, choosing $s \in (0.5,1)$, the probability that $\overline{\mu}_Y$ fails is at most:
	\begin{align*}
		\sum_{i=0}^{\lfloor r/2 \rfloor} \binom{r}{i} (s)^{i} (1-s)^{r-i} &\leq\sum_{i=0}^{\lfloor r/2 \rfloor} \binom{r}{i} (s)^{r/2} (1-s)^{r/2}\\
		&\leq (s-s^2)^{r/2}\sum_{i=0}^{\lfloor r/2 \rfloor} \binom{r}{i}\\
		&\leq (s-s^2)^{r/2}\cdot 2^{-r}\\
		&\leq (4s-4s^2)^{r/2}
	\end{align*}
	We use the previous bound to choose $r$ such that $(4s-4s^2)^{r/2}\leq \delta$. In particular, for $s=3/4$, we find:
	\begin{equation*}
		r = \frac{2}{\log(4/3)}\log(1/\delta)
	\end{equation*}
	To recap: construct a single experiment $Y_n$ using $n=O(\sigma^2_Y/\mu^2_Y\epsilon^{-2})$ samples, repeat the experiment $r=O(\log1/\delta)$ times, and return median $\overline{\mu}_Y$. Using $O(\sigma^2_Y/\mu^2_Y\epsilon^{-2}\log1/\delta)$ samples, we showed this procedure returns $\overline{\mu}_Y$ in the range $(1\pm\epsilon)\cdot\mu_Y$ with at least probability $1-\delta$.
\end{proof}
\section*{Acknowledgments}
The proof of Theorem~\ref{theom:markov} is adapted from Prof. Sinclair's online lecture notes~\cite{Sinclair2018}.

\end{document}